\documentclass[10pt]{article}
\usepackage[margin=1in]{geometry}
\usepackage[utf8]{inputenc}
\usepackage{microtype}
\usepackage[all=normal,bibliography=tight]{savetrees}

\usepackage{listings}
\usepackage{cite}

\usepackage{amstext,amsfonts,amsthm,amsmath,amssymb}
\usepackage{amssymb}

\usepackage{graphicx,tikz}
\usepackage{comment}
\usepackage{url}
\usepackage{xspace}
\usepackage{todonotes}
\usepackage[shortlabels]{enumitem}

\usepackage[absolute]{textpos}

\usepackage{caption}
\usepackage{subcaption}
\usetikzlibrary{decorations.markings}

\def\cqedsymbol{\ifmmode$\lrcorner$\else{\unskip\nobreak\hfil
\penalty50\hskip1em\null\nobreak\hfil$\lrcorner$
\parfillskip=0pt\finalhyphendemerits=0\endgraf}\fi} 

\newcommand{\cqed}{\renewcommand{\qed}{\cqedsymbol}}

\newtheorem{lemma}{Lemma}[section]

\newtheorem{corollary}[lemma]{Corollary}
\newtheorem{theorem}[lemma]{Theorem}
\newtheorem{claim}[lemma]{Claim}
\theoremstyle{definition}
\newtheorem{definition}[lemma]{Definition}

\newcommand{\wei}{\omega}
\newcommand{\Oh}{\mathcal{O}}
\newcommand{\tree}{\mathcal{T}}
\newcommand{\brleaf}{\zeta}
\newcommand{\brnoose}{\gamma}

\newcommand{\brcut}{\mathrm{mid}}
\newcommand{\Aa}{\mathcal{A}}
\newcommand{\Bb}{\mathcal{B}}

\newcommand{\Pp}{\mathcal{P}}
\newcommand{\Qq}{\mathcal{Q}}

\newcommand{\tspname}{\textsc{Directed Subset TSP}}

\newcommand{\Gbunch}{G^\Theta}
\newcommand{\Wbunch}{W^\Theta}
\newcommand{\Pbunch}{P_G^\Theta}
\newcommand{\Ggrid}{G^\square}
\newcommand{\Wgrid}{W^\square}
\newcommand{\grid}{\Gamma}
\newcommand{\Pgrid}{P_G^\square}

\newcommand{\Tedges}{F}

\renewcommand{\leq}{\leqslant}

\renewcommand{\geq}{\geqslant}
\renewcommand{\ge}{\geqslant}

\title{A subexponential parameterized algorithm for Directed Subset Traveling Salesman Problem on planar graphs\thanks{%
The results of this paper have been presented in an extended abstract at FOCS 2018~\cite{MarxPP18}.
This research is a part of projects that have received funding from the European Research Council (ERC) under the European Union's Horizon 2020 research and innovation programme
under grant agreements No.~280152 and 725978 (D\'{a}niel Marx) and 714704 (Marcin Pilipczuk).
The research of Micha\l{} Pilipczuk is supported by Polish National Science Centre grant UMO-2013/11/D/ST6/03073.
Micha\l{} Pilipczuk is also supported by the Foundation for Polish Science (FNP) via the START stipend programme.}}
\author{ 
  D\'aniel Marx\thanks{%
  CISPA Helmholtz Center for Information Security, Saarland Informatics Campus, Germany (\texttt{marx@cispa.de})}
  \and 
  Marcin Pilipczuk\thanks{
    Institute of Informatics, University of Warsaw, Poland (\texttt{marcin.pilipczuk@mimuw.edu.pl}).
  }
  \and 
  Micha\l{} Pilipczuk\thanks{
    Institute of Informatics, University of Warsaw, Poland (\texttt{michal.pilipczuk@mimuw.edu.pl}).
  }
}

\date{}

\begin{document}

\maketitle

\begin{textblock}{20}(0, 13.0)
\includegraphics[width=40px]{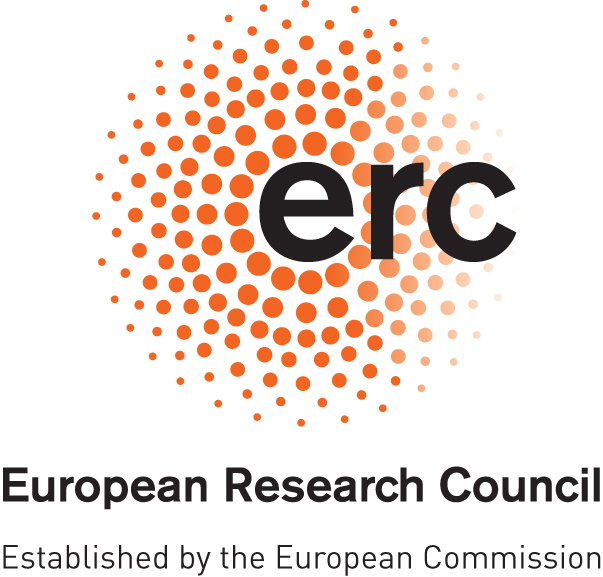}%
\end{textblock}
\begin{textblock}{20}(0, 13.8)
\includegraphics[width=40px]{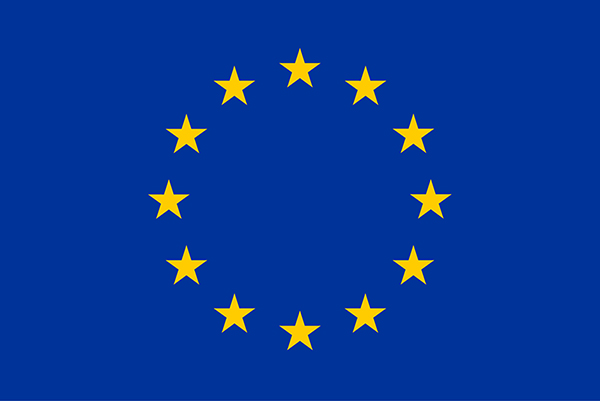}%
\end{textblock}

\begin{abstract}
There are numerous examples of the so-called ``square root phenomenon'' in the field of parameterized algorithms: many of the most fundamental graph problems, parameterized by some natural parameter $k$, become significantly simpler when restricted to planar graphs and in particular the best possible running time is exponential in $\Oh(\sqrt{k})$ instead of $\Oh(k)$ (modulo standard complexity assumptions). We consider a classic optimization problem \textsc{Subset Traveling Salesman}, where we are asked to visit all the terminals $T$ by a minimum-weight closed walk.
We investigate the parameterized complexity of this problem in planar graphs, where the number $k=|T|$ of terminals is
regarded as the parameter. We show that 
{\sc{Subset TSP}} can be solved in time $2^{\Oh(\sqrt{k}\log k)}\cdot n^{\Oh(1)}$ even on edge-weighted directed planar graphs. This improves upon the algorithm of Klein and Marx [SODA 2014] with the same running time that worked only on undirected planar graphs with polynomially large integer weights. 
\end{abstract}

\section{Introduction}\label{sec:intro}
It has been observed in the context of different algorithmic paradigms
that planar graphs enjoy important structural properties that allow
more efficient solutions to many of the classic hard algorithmic
problems. The literature on approximation algorithms contains many examples of optimization
problems that are APX-hard on general graphs, but admit
polynomial-time approximation schemes (PTASes) when restricted to planar
graphs (see, e.g., \cite{DBLP:journals/talg/BorradaileK16,DBLP:conf/stoc/FoxKM15,DBLP:conf/stacs/KleinMZ15,DBLP:conf/soda/EisenstatKM14,DBLP:conf/soda/EisenstatKM12,DBLP:conf/soda/BateniHKM12,DBLP:journals/talg/BorradaileKM09,DBLP:journals/siamcomp/Klein08,DBLP:conf/stoc/BateniDHM16,DBLP:journals/jacm/BateniHM11}). When looking for exact solutions, even though the planar versions of most NP-hard problems
remain NP-hard, a more
fine-grained look reveals that significantly better running times are
possible for planar graphs. As a typical example, consider the
\textsc{3-Coloring} problem: it can be solved in time $2^{\Oh(n)}$ in
general graphs and, assuming the Exponential-Time Hypothesis (ETH),
this is best possible as there is no $2^{o(n)}$-time
algorithm. However, when restricted to planar graphs,
\textsc{3-Coloring} can be solved in time $2^{\Oh(\sqrt{n})}$, which is
again best possible assuming ETH: the existence of a $2^{o(\sqrt{n})}$-time algorithm would contradict ETH.
(A detailed discussion on these and similar results can be found in Section~14.2 of~\cite{platypus}.)
There are many other problems
that behave in a similar way and this can be attributed to the
combination of two important facts: (1) every planar graph on $n$
vertices has treewidth $\Oh(\sqrt{n})$ and (2) given an $n$-vertex graph
of treewidth $t$, most of the natural combinatorial problems can be
solved in time $2^{\Oh(t)}\cdot n^{\Oh(1)}$ (or perhaps $2^{\Oh(t\cdot \textup  {polylog } t)} \cdot n^{\Oh(1)}$). On the lower bound side, to rule out $2^{o(\sqrt{n})}$-time algorithms, it is sufficient to observe that most planar NP-hardness proofs increase the size of the instance at most quadratically (because of the introduction of crossing gadgets). For example, there is a reduction
that given an instance of \textsc{3SAT} with $n$ variables and $m$ clauses produce 
an instance of \textsc{3-Coloring} that is a planar graph with $\Oh((n+m)^2)$ vertices. Together with ETH, such a reduction rules out $2^{o(\sqrt{n})}$-time algorithms for planar \textsc{3-Coloring}. Thus the existence of this ``square root phenomenon'' giving $2^{\Oh(\sqrt{n})}$ time complexity is well-understood both from the algorithmic and complexity viewpoints.

Our understanding of this phenomenon is much less complete for
parameterized problems. A large fraction of natural fixed-parameter
tractable graph problems can be solved in time $2^{\Oh(k)}\cdot
n^{\Oh(1)}$ (with notable exceptions
\cite{DBLP:journals/siamcomp/CyganPP16,DBLP:journals/siamcomp/LokshtanovMS18})
and a large fraction of W[1]-hard problems can be solved in time
$n^{\Oh(k)}$. There are tight or almost-tight lower bounds showing the
optimality of these running times. By now, there is a growing list of
problems where the running time improves to
$2^{\Oh(\sqrt{k}\cdot\textup{polylog }k)}\cdot n^{\Oh(1)}$ or to
$n^{\Oh(\sqrt{k}\cdot\textup{polylog }k)}$ when restricted to planar
graphs. For a handful of problems (e.g., \textsc{Independent Set},
\textsc{Dominating Set}, \textsc{Feedback Vertex Set},
\textsc{$k$-Path}) this improvement can be explained in a compact way
by the elegant theory of bidimensionality \cite{DemaineFHT05}.
However, there is no generic argument (similar to the simple argument
described above for the existence of $2^{\Oh(\sqrt{n})}$ algorithms)
why such an improvement should be possible for most parameterized
problems. The fact that every $n$-vertex planar graph has treewidth
$\Oh(\sqrt{n})$ does not seem to help in improving the $2^{\Oh(k)}$
factor to $2^{\Oh(\sqrt{k})}$ in the running time. The algorithmic
results of this form are thus very problem-specific, exploiting
nontrivial observations on the structure of the solution or invoking
other tools tailored to the problem's nature. Recent results include
algorithms for \textsc{Subset TSP} \cite{KleinM14}, \textsc{Multiway
  Cut} \cite{KleinM12,Marx12}, unweighted \textsc{Steiner Tree}
parameterized by the number of edges of the solution
\cite{pst-kernel,PilipczukPSL13}, \textsc{Strongly Connected Steiner
  Subgraph} \cite{ChitnisHM14}, \textsc{Subgraph Isomorphism}
\cite{FominLMPPS16}, facility location problems
\cite{DBLP:conf/esa/MarxP15}, \textsc{Odd Cycle Transversal}
\cite{DBLP:conf/fsttcs/LokshtanovSW12}, and \textsc{3-Coloring}
parameterized by the number of vertices with degree $\ge 4$
\cite{DBLP:journals/corr/AboulkerBHMT15}.

      It is plausible to expect that other natural problems also have
      significantly faster parameterized algorithms on planar
      graphs. The reason for this optimism is twofold. First, even
      though the techniques used to obtain the results listed above
      are highly problem-specific, they suggest that planar graphs
      have rich structural properties, connected to the existence of sublinear separators, that can be exploited in various ways and in multiple settings. Second, lower bounds ruling out subexponential algorithms for
      planar problems intuitively require large expressive power of the combinatorics of the problem at hand, which is lacking in the case most natural problems. More precisely, to prove that a parameterized algorithm with running time $2^{o(k)}\cdot n^{\Oh(1)}$ violates ETH, one needs to give a reduction from \textsc{3SAT} with $m$ clauses to a planar instance with parameter $k=\Oh(m)$. However, in a typical reduction for a typical problem, the output planar graph has $\Omega(m^2)$ ``crossing gadgets'', each increasing the parameter, which ultimately yields $k = \Omega(m^2)$. 
      
      The intuition presented in the paragraph above is, however, not quite right.
      In a very recent result, we have found a novel type of reduction that gets around the discussed limitations and, assuming ETH, rules out the existence of $2^{o(k)} \cdot n^{\Oh(1)}$-time algorithms for \textsc{Steiner Tree} parameterized by the number of terminals~\cite{MarxPP18}.
 A result of similar flavor has been reported by Bodlaender et al.~\cite{BodlaenderNZ16}, who, under the same assumption, ruled out the existence of a $2^{o(k / \log k)} \cdot n^{\Oh(1)}$-time
 algorithm for \textsc{Subgraph Isomorphism} (and a few related problems) in planar graphs, parameterized by the size of the pattern graph. These results put the search for subexponential parameterized algorithms in planar graphs in a new perspective, as they show that the boundary between subexponential tractability and intractability is much more wild --- and therefore interesting --- than previously expected.

      \paragraph*{Our contribution.} In this paper we address a classic
      problem on planar graphs for which the existence
      of subexponential parameterized algorithm was open. Given a
      graph $G$ with a subset $T$ of vertices distinguished as
      terminals, the \textsc{Subset TSP} problem asks for a shortest
      closed walk visiting the terminals in any order. Parameterized by the
      number $k=|T|$ of terminals, the problem is fixed-parameter
      tractable in arbitrary graphs: it can be solved in time
      $2^k\cdot n^{\Oh(1)}$ by first computing the distance between
      every pair of terminals, and then solving the resulting $k$-terminal
      instance using the standard Bellman-Held-Karp dynamic programming algorithm. Klein and Marx \cite{KleinM14}
      showed that if $G$ is an undirected planar graph with
      polynomially bounded edge weights, then the problem can be
      solved significantly faster, in time $2^{\Oh(\sqrt{k}\log k)}\cdot
      n^{\Oh(1)}$. The limitations of polynomial weights and
      undirected graphs are inherent to this algorithm: it starts with
      computing a locally 4-step optimal solution (which requires
      polynomial weights to terminate in polynomial time) and relies
      on an elaborate subtour-replacement argument (which breaks down if
      the tour has an orientation). The main argument is the unexpected 
      claim that the union of an optimal and a
      locally 4-step optimal tour has treewidth $\Oh(\sqrt{k})$. 
      
      Our
      result is a more robust and perhaps less surprising
      algorithm that achieves the same running time, but does not
      suffer from these limitations.

\begin{theorem}\label{thm:dirtsp}
  Given an edge-weighted directed planar graph $G$ with terminals $T$,
  \textsc{Subset TSP} parameterized by $k=|T|$ can be solved in time
  $2^{\Oh(\sqrt{k}\log k)} n^{\Oh(1)}$.
\end{theorem}
The similarity of \textsc{Subset TSP} and \textsc{Steiner Tree}, for which a lower bound ruling out $2^{o(k)} n^{\Oh(1)}$ time algorithms in planar graphs has been recently shown~\cite{MarxPP18}, 
    suggests a very intricate boundary between parameterized problems that admit and do not admit subexponential parameterized algorithms in planar graphs.

The proof of Theorem~\ref{thm:dirtsp} has the same high-level idea as the algorithm of Klein and Marx \cite{KleinM14}: a family of $2^{\Oh(\sqrt{k}\log k)}$ subsets of terminals is computed,
followed by applying a variant of the  Bellman-Held-Karp dynamic programming algorithm that considers only subsets of terminals that appear in this family. However, the way we compute 
such a family is very different: the construction of Klein and Marx \cite{KleinM14} crucially relies on how the optimal solution interacts with the locally 4-step optimal solution (e.g., they cross each other $\Oh(k)$ times), while our argument here does not use any such assumption. For directed graphs, we can extract much fewer properties of the structure of the solution or how it interacts with some other object. For example, we cannot require that the optimum solution is non-self-crossing and the number of self-crossings cannot be even bounded by a function of $k$. Thus in order to find an algorithm working on directed graphs, we need to use more robust algorithmic ideas that better explain why it is possible to have subexponential parameterized algorithms for this problem.

In Section~\ref{sec:overview}, we highlight these new ideas 
in an overview of the algorithm of Theorem~\ref{thm:dirtsp}.
After brief preliminaries in Section~\ref{sec:prelims} and 
an auxiliary result on noose enumeration in Section~\ref{sec:nooses},
    we prove Theorem~\ref{thm:dirtsp} in Section~\ref{sec:tsp}.

\section{An overview of the algorithm}\label{sec:overview}

In this section we give an overview of the approach leading to the subexponential parameterized algorithm for \tspname{}, that is, the proof of Theorem~\ref{thm:dirtsp}.
We first describe the high-level strategy of restricting the standard dynamic programming algorithm to a smaller family of candidate states.
Then we explain the main idea of how such a family of candidate states can be obtained; however, we introduce multiple simplifying assumptions and hide most of the technical problems.
Finally, we briefly review the issues encountered when making the approach work in full generality, and explain how we cope with them.
We strongly encourage the reader to read this section before proceeding to the formal description, as in the formal layer many of the key ideas become somehow obfuscated by the technical details surrounding them.

\subsection{Restricted dynamic programming}

Restricting dynamic programming to a small family of candidates states is by now a commonly used technique in parameterized complexity.
The idea is as follows.
Suppose that we search for a minimum-cost solution to a combinatorial problem, and this search
can be expressed as solving a number of subproblems in a dynamic programming fashion, where each subproblem corresponds
to a {\em{state}} from a finite state space $\cal S$. Usually, subproblems correspond to partial solutions, and transitions between states correspond to extending one partial solution to a larger
partial solution at some cost, or combining two or more partial solutions to a larger one. 
For simplicity, assume for now that we only extend single partial solutions to larger ones, rather than combine multiple partial solutions.
Then the process of assembling the final solution from partial solutions may be described as a nondeterministic algorithm that guesses consecutive extensions, leading
from a solution to the most basic subproblem to the final solution for the whole instance. 
The sequence of these extensions is a path (called also a {\em{computation path}}) in a directed graph on $\cal S$ where the transitions between the states are the arcs.
Then the goal is to find a minimum-weight path from the initial state to any final state, which can be done in time linear in the size of this state graph, provided it is acyclic.

In order to improve the running time of such an algorithm one may try the following strategy. Compute a subset of states $\cal S'\subseteq \cal S$ with the following guarantee: there is a computation path leading
to the discovery of a minimum-weight solution that uses only states from $\cal S'$. Then we may restrict the search only to states from $\cal S'$. So the goal is to find a subset of states $\cal S'$ that
is rich enough to ``capture'' some optimum solution, while at the same time being as small as possible so that the algorithm is efficent.

Let us apply this principle to \tspname{}. Consider first the most standard dynamic programming algorithm for this problem, working on general graphs in time $2^{k}\cdot n^{\Oh(1)}$,
where we denote $k=|T|$ by convention. 
Each subproblem is described by a subset of terminals $S\subseteq T$ and two terminals $s_1,s_2\in S$. The goal in the subproblem is to find the shortest tour that starts in $s_1$, ends in $s_2$, and visits
all terminals of $S$ along the way. The transitions are modelled by a possibility of extending a solution for the state $(S,s_1,s_2)$ to a solution for the state $(S\cup \{s'\},s_1,s')$ for any $s'\notin S$ at the cost of adding the shortest path
from $s_2$ to $s'$. The minimum-weight tour can be obtained by taking the best among solutions obtained as follows: for any $s_1,s_2\in T$, take the solution for the subproblem $(T,s_1,s_2)$ and augment it by adding
the shortest path from $s_2$ to $s_1$. Observe that the above algorithm is essentially
the standard Bellman-Held-Karp dynamic programming algorithm for \textsc{TSP}, applied
to the shortest path metric on $T$.

This is not the dynamic programming algorithm we will be improving upon. The reason is that restricting ourselves to constructing one interval on the tour at a time makes it difficult to enumerate a small
subfamily of states capturing an optimum solution. Also, the above dynamic programming algorithm computes an optimum partial solution to every subproblem. In our dynamic programming algorithm
we will be only able to ensure optimality for states appearing on the chosen computation path for some chosen optimal solution.

Instead, we consider a more involved variant of the above dynamic programming routine, which intuitively keeps track of $\Oh(\sqrt{k})$ intervals on the tour at a time.
More precisely, each subproblem is described by a state defined as a pair $(S,\mathcal{M})$, 
where $S\subseteq T$ is a subset of terminals to be visited, and $\mathcal{M}$ (also called {\em{connectivity pattern}}) is a set of pairwise disjoint pairs of terminals from $S$, where
$|\mathcal{M}|\leq C\sqrt{k}$ for some universal constant $C$.
The goal in the subproblem is to compute a family of paths $\mathcal{P}_{(S,\mathcal{M})}$ of minimum possible weight having the following properties:
for each $(s_1,s_2)\in \mathcal{M}$ there is a path in $\mathcal{P}_{(S,\mathcal{M})}$ that leads from $s_1$ to $s_2$, and each terminal from $S$ lies on some path in $\mathcal{P}_{(S,\mathcal{M})}$.
Note, however, that we do not specify, for each terminal from $S$, on which of the paths it has to lie.

Solutions to such subproblems may be extended by single terminals as in the standard dynamic programming, but they can be also combined in pairs.
More precisely, consider two solutions $\Pp_1$ and $\Pp_2$ respectively for $(S_1,\mathcal{M}_1)$ and $(S_2,\mathcal{M}_2)$ where $S_1\cap S_2=\emptyset$.
For $i=1,2$, let $X_i$ and $Y_i$ be the starting and the ending terminals of the matching of $\mathcal{M}_i$.
Let $X \subseteq X_1 \cup X_2$ and $Y \subseteq Y_1 \cup Y_2$ be two equal-sized sets, let $X' = (X_1 \cup X_2) \setminus X$ and $Y' = (Y_1 \cup Y_2) \setminus Y$; note that $|X| = |Y|$ implies $|X'| = |Y'|$.
Let $\mathcal{M}'$ be a matching between $Y'$ and $X'$ and let $\Pp'$ be the family of shortest paths between the pairs in $\mathcal{M}'$. 
Then $\Pp := \Pp_1 \cup \Pp' \cup \Pp_2$ is a family of walks starting in $X$ and ending in $Y$ plus possibly some closed walks. 
If $\Pp$ contains no closed walks and $\mathcal{M}$ is a matching between $X$ and $Y$ matching starting and ending terminals of $\Pp$, then $\Pp$ is a candidate solution 
to $(S_1 \cup S_2,\mathcal{M})$. The dynamic programming algorithm is able to choose the minimum-weight solution to $(S_1 \cup S_2, \mathcal{M})$ obtained
for different choices of $(S_1,\mathcal{M}_1)$, $(S_2,\mathcal{M}_2)$, and $\mathcal{M}'$ (which induces the choice of $X'$, $Y'$, $X$, and $Y$). 

Since 
we assume that $|\mathcal{M}_1|,|\mathcal{M}_2|,|\mathcal{M}|\leq C\sqrt{k}$, there are only $k^{\Oh(\sqrt{k})}$ ways to perform a merge as in the previous paragraph. 
While this dynamic programming formally does not conform to the ``linear view'' described in the paragraphs above, as it may merge partial solutions for two simpler states into a larger partial solution,
it is straightforward to translate the concept of restricting the state space to preserve the existence of a computation path (here, rather a computation tree) leading 
to a minimum-cost solution.

Observe that since in a state $(S,\mathcal{M})$ we stipulate the size of $\mathcal{M}$ to be $\Oh(\sqrt{k})$, the total number of states with a fixed subset $S\subseteq T$ is $k^{\Oh(\sqrt{k})}$.
Thus, from the discussion above we may infer the following lemma, stated here informally.

\begin{lemma}[Lemma~\ref{lem:dp}, informal statement]\label{inflem:dp}
Let $(G,T)$ be an instance of \tspname{}.
Suppose we are also given a family $\Bb$ of subsets of $T$ with the following guarantee: 
there is a computation path of the above dynamic programming leading to an optimum solution that uses only states of the form $(S,\mathcal{M})$ where $S\in \Bb$.
Then we can find an optimum solution for the instance $(G,T)$ in time $k^{\Oh(\sqrt{k})}\cdot (|\Bb|\cdot |G|)^{\Oh(1)}$.
\end{lemma}

Concluding, we are left with constructing a family $\Bb$ of subsets of $T$ that satisfies the prerequisites of Lemma~\ref{inflem:dp} and has size $k^{\Oh(\sqrt{k})}$, provided the underlying
graph $G$ is planar. For this, we will crucially use topological properties of $G$ given by its planar embedding.

\subsection{Enumerating candidate states}\label{sec:enumeration}

Suppose $(G,T)$ is the input instance of \tspname{} where $G$ is planar.
Without loss of generality we may assume that $G$ is strongly connected.
Fix some optimum solution $W$, which is a closed walk in the input graph $G$ that visits every terminal. 

\paragraph*{Simplifying assumptions.} We now introduce a number of simplifying assumptions.
These assumptions are made with loss of generality, and we introduce them in order to present our main ideas in a setting that is less obfuscated by technical details.
\begin{enumerate}[({A}1)]
\item Walk $W$ is in fact a simple directed cycle, without any self-intersections. In particular, the embedding of $W$ in the plane is a closed curve without self-intersections; denote this curve by~$\delta$.
\item The walk $W$ visits every terminal exactly once, so that we may speak about the (cyclic) order of visiting terminals on $W$.
\end{enumerate}
Note that Assumption A2 follows from A1, but we prefer to state them separately as later we first obtain Assumption A2
and then discuss Assumption A1.

We will also assume that shortest paths are unique in $G$, but this can be easily achieved by perturbing the weights of edges of $G$ slightly.

Suppose now that we have another closed curve $\gamma$ in the plane, without self-intersections, that crosses $\delta$ in $p=\Oh(\sqrt{k})$ points, none of which is a terminal. 
Curve $\gamma$ divides the plane into two open regions (maximal connected parts of the plane after removal of $\gamma$)---say $R_1,R_2$---and thus $\delta$ is divided into 
$p$ intervals which are alternately contained in $R_1$ and $R_2$. Let $S$ be the set of terminals visited on the intervals contained in $R_1$. 
Then it is easy to see that $S$ is a good candidate for a subset of terminals
that we are looking: $S$ forms at most $\Oh(\sqrt{k})$ contiguous intervals in the order of visiting terminals by $W$, and hence for the connectivity pattern $\mathcal{M}$
consisting of the first and last terminals on these intervals, 
the state $(S,\mathcal{M})$ would describe a subproblem useful for discovering $W$ as the part of $W$ inside $\gamma$ is a solution to this state.

However, we are not really interested in capturing one potentially useful state, but in enumerating a family of candidate states that contains a complete computation path leading to the discovery of an optimum solution.
Hence, we rather need to capture a hierarchical decomposition of $T$ using curves $\gamma$ as above, so that terminal subsets $S$ induce the sought computation path.
For this, we will use the notion of {\em{sphere-cut decompositions}} of planar graphs, and the well-known fact that every $k$-vertex planar graph admits a sphere-cut decomposition of width $\Oh(\sqrt{k})$.

\paragraph*{Sphere-cut decompositions.}
A {\em{branch decomposition}} of a graph $G$ is a ternary tree $\tree$ (i.e. one with every internal node of degree $3$), together with a bijection $\brleaf$ between the leaves of $\tree$ and the edges of $G$.
For every edge $e$ of $\tree$, the removal of $e$ from $\tree$ splits $\tree$ into two subtrees, say $\tree^1$ and $\tree^2$. 
The {\em{cut}} (or {\em{middle set}}) of $e$, 
denoted $\brcut(e)$, is the set of those vertices of $G$ that are incident to both an edge corresponding (via $\brleaf$) to a leaf contained in $\tree_1$, and to an edge corresponding
to a leaf contained in $\tree_2$.
The {\em{width}} of a branch decomposition $(\tree,\brleaf)$ is the maximum size of a cut in it. The {\em{branchwidth}} of a graph $G$ is the minimum possible width of a branch decomposition
of $G$. It is well-known that a planar graph on $k$ vertices has branchwidth $\Oh(\sqrt{k})$ (see e.g.~\cite{FominT06}).

After rooting a branch decomposition $(\tree,\brleaf)$ in any node, it can be viewed as a hierarchical decomposition of the edge set of $G$ using vertex cuts of size bounded by the width of the decomposition.
Seymour and Thomas~\cite{SeymourT94} proved that in plane graphs we can always find an optimum-width branch decomposition that somehow respects the topology of the plane embedding of a graph.
Precisely, having fixed a plane embedding of a connected graph $G$, call a closed curve $\gamma$ in the plane a {\em{noose}} if $\gamma$ has no self-intersections and it crosses the embedding of $G$ only at 
vertices\footnote{In standard literature, e.g.~\cite{SeymourT94}, a noose is moreover required to visit every face of $G$ at most once; 
in this paper we do not impose this restriction.}; in particular it does not intersect any edge of $G$. Such a curve $\gamma$ divides the plane into two regions, which naturally induces a partition of the edge set of $G$ into edges that are embedded in the first,
respectively the second region. A {\em{sphere-cut decomposition}} of $G$ is a branch decomposition $(\tree,\brleaf)$ where in addition every edge $e$ of $\tree$ is assigned its noose $\gamma(e)$ such that
$\gamma(e)$ traverses the vertices of $\brcut(e)$ and 
the partition of the edge set induced by $\gamma(e)$ corresponds (via $\brleaf$) to the partition of the leaf set of $\tree$ induced by removing $e$ from $\tree$.
Then the result of Seymour and Thomas~\cite{SeymourT94} may be stated as follows:
every connected planar graph has a sphere-cut decomposition of width equal to its 
branchwidth\footnote{In~\cite{SeymourT94} it is also assumed that the graph is bridgeless, which corresponds to the requirement that every face is visited by a noose at most once. 
It is easy to see that in the absence of this requirement it suffices to assume the connectivity of the graph.}. Together with the square-root behavior of the branchwidth of a planar graph, this implies the following.

\begin{theorem}[see e.g.~\cite{FominT06}]
Every connected plane graph that has at most $k$ vertices of degree at least $3$
has a sphere-cut decomposition of width at most $\alpha\sqrt{k}$, for some constant $\alpha$.
\end{theorem}

Turning back to our \tspname{} instance $(G,T)$ and its optimum solution $W$, our goal is to enumerate a possibly small family of subsets of $T$ that contains some complete computation path leading to the discovery
of $W$. The remainder of the construction is depicted in Figure~\ref{fig:overview} (on page \pageref{fig:overview}) and we encourage the reader to analyze it while reading the description. The description is divided into ``concepts'', 
which are not steps of the algorithm, but of the analysis leading to its formulation.

\paragraph*{Concept 1: adding a tree.}
Take any (inclusionwise) minimal tree $H_0$ in the underlying undirected graph of $G$ spanning all terminals of $T$.
We may assume that $H_0$ contains at most $k$ leaves that are all terminals, at most $k-2$ vertices of degree at least $3$, and otherwise it consists of at most $2k-3$ simple paths 
connecting these leaves and vertices of degree at least $3$ (further called {\em{special}} vertices of $H_0$).
To avoid technical issues and simplify the picture, we introduce another assumption.
\begin{enumerate}[({A}3)]
\item Walk $W$ and tree $H_0$ do not share any edges.
\end{enumerate}

Let $H$ be the graph formed by the union of $W$ and $H_0$. Even though both $W$ and $H_0$ consist of at most $2k$ simple paths in $G$, the graph $H$ may have many vertices of degree more than $3$. 
One of the possible scenarios for that is when a subpath $Q$ between two consecutive terminals on $W$ and a path $P$ in $H_0$ that connects two special vertices of $H_0$  cross many times.
The intuition is, however, that the planar structure of $H$ roughly resembles a structure of a planar graph on $\Oh(k)$ vertices, and a sphere-cut decomposition of this planar graph of width $\Oh(\sqrt{k})$ 
should give rise to the sought hierarchical partition of terminals leading to the discovery of $W$ by the dynamic programming algorithm.

Another way of looking at the tree $H_0$ is that we can control the homotopy types of closed curves in the plane punctured at the terminals, by examining how they cross with $H_0$.

Let us remark that, of course, the definition of the graph $H$ relies on the (unknown to the algorithm) solution $W$, though the tree $H_0$ can be fixed and used by the algorithm.
At the end we will argue that having fixed $H_0$, we may enumerate a family of $k^{\Oh(\sqrt{k})}$ candidates for nooses in a sphere-cut decomposition of $H$. 
Roughly, for each such noose $\gamma$ we consider the bi-partition of terminals according to the regions of the plane minus $\gamma$ in which they lie, 
and we put all terminal subsets constructed in this manner into a family $\Bb$, which
is of size $k^{\Oh(\sqrt{k})}$. Then restricting the dynamic programming algorithm to $\Bb$ as in Lemma~\ref{inflem:dp} gives us the required time complexity.

\paragraph*{Concept 2: Contracting subpaths of $W$.}
Hence, the goal is to simplify the structure of $H$ so that it admits a sphere-cut decomposition of width $\Oh(\sqrt{k})$.
Consider any pair of terminals $t_1,t_2$ visited consecutively on $W$, and let $P$ be the subpath of $W$ from $t_1$ to $t_2$. Consider contracting all internal vertices on $P$ into a single vertex, thus
turning $P$ into a path $P'$ on $2$ edges and $3$ vertices. Let $H'$ be the graph obtained from $H$ by contracting each path between two consecutive terminals on $W$ in the manner described above. Observe that
thus, $H'$ has less than $3k$ vertices of degree at least $3$: there are at most $2k$ vertices on the contracted $W$ in total, and there can be at most $k-2$ vertices of degree at least $3$
on $H_0$ that do not lie on $W$. 
Then $H'$ has a sphere-cut decomposition of width $\leq \alpha\sqrt{3k}$, say $(\tree,\eta,\gamma(\cdot))$.

Consider the family $\mathcal{D}$ of subsets of terminals constructed as follows. For each noose $\gamma(e)$ for $e\in \tree$, that is, appearing in the sphere-cut decomposition $(\tree,\eta,\gamma(\cdot))$,
and each partition $(X,Y)$ of terminals traversed by $\gamma(e)$ (there are at most $\alpha\sqrt{3k}$ such terminals, so $2^{\Oh(\sqrt{k})}$ such partitions),
add to $\mathcal{D}$ the following two terminal subsets: the set of terminals enclosed by $\gamma(e)$ plus $X$, and the set of terminals excluded by $\gamma(e)$ plus $Y$. 
It can be now easily seen that $\mathcal{D}$ contains a complete
computation path that we are looking for, as each terminal subset included in $\mathcal{D}$ forms at most $\Oh(\sqrt{k})$ contiguous intervals in the cyclic order of terminals on $W$, 
and the decomposition tree $\tree$
shows how our dynamic programming should assemble subsets appearing in $\mathcal{D}$ in pairs up to the whole terminal set. In other words, if we manage to construct a family $\mathcal{B}$
of size $k^{\Oh(\sqrt{k})}$ with a guarantee that it contains the whole $\mathcal{D}$, then we will be done by Lemma~\ref{inflem:dp}.

\paragraph*{Concept 3: Enumeration by partial guessing.}
Obviously, the graph $H'$ is not known to the algorithm, as its definition depends on the fixed optimum solution $W$.
Nevertheless, we may enumerate a reasonably small family of candidates for nooses used in its sphere-cut decomposition $(\tree,\eta,\gamma(\cdot))$.
The main idea is that even though the full structure of $H'$ cannot be guessed at one shot within $k^{\Oh(\sqrt{k})}$ possibilities, 
each noose we are interested in traverses only at most $\alpha\sqrt{3k}$ vertices of $H'$, and hence it is sufficient to guess only this small portion of $H'$.

More precisely, let $Q$ be the subset of those vertices of $H'$ that are obtained from contracting the subpaths of $W$ between consecutive terminals.
Fix a noose $\gamma$ appearing in the sphere-cut decomposition of $H'$, that is, $\gamma=\gamma(e)$ for some $e\in \tree$.
Then $\gamma$ traverses at most $\alpha\sqrt{3k}$ vertices of $Q$; say that $R\subseteq Q$ is the set of these vertices.
We can now enumerate a set of $k^{\Oh(\sqrt{k})}$ candidates for $\gamma$ by performing the following steps (by {\em{guessing}} we mean iterating through all options):
\begin{enumerate}[(1)]
\item Guess a set $\mathcal{R}$ of at most $\alpha\sqrt{3k}$ pairs of distinct terminals.
\item For each $(s,t)\in \mathcal{R}$, take the shortest path $P_{(s,t)}$ from $s$ to $t$ and consider contracting it to a single vertex $p_{(s,t)}$.
\item Take the fixed tree $H_0$ that spans terminals in $G$, apply the above contractions in $G$, and let $H_{\mathcal{R}}$ be the graph to which $H_0$ is transformed under these contractions.
\item Enumerate all nooses $\gamma$ that meet $H_{\mathcal{R}}$ only at terminals and vertices of degree at least $3$, and traverse at most $\alpha\sqrt{3k}$ such vertices.
\end{enumerate}
In Step 1 we have at most $k^{\Oh(\sqrt{k})}$ options for such a set $\mathcal{R}$, and the contractions in Steps~2 and~3 turn $H_0$ into a planar graph $H_{\mathcal{R}}$ with $\Oh(k)$ vertices.
It is not hard to convince oneself that in such a graph, there are only $k^{\Oh(\sqrt{k})}$ nooses satisfying the property expressed in the Step 4, so all in all we enumerate at most
$k^{\Oh(\sqrt{k})}$ curves in the plane, each traversing at most $\alpha\sqrt{3k}$ terminals.
Now, for each enumerated curve $\gamma$, we include into $\Bb$ two terminal subsets: the set of terminals enclosed by $\gamma$ and the set of terminals excluded by $\gamma$.
Thus $|\Bb|=k^{\Oh(\sqrt{k})}$.

It remains to argue that $\Bb$ contains the whole family $\mathcal{D}$ that was given by the sphere-cut decomposition $(\tree,\eta,\gamma(\cdot))$ of $H'$, so that Lemma~\ref{inflem:dp} may be applied.
It should be quite clear that it is sufficient to show that every noose $\gamma$ appearing in $(\tree,\eta,\gamma(\cdot))$ is enumerated in Step 4 of the procedure from the previous paragraph.
However, nooses with respect to $H_{\mathcal{R}}$ are formally not necessarily nooses with respect to $H'$, as we wanted. Nevertheless, if a noose $\gamma$ appears in the sphere-cut decomposition
$(\tree,\eta,\gamma(\cdot))$ of $H'$, and we take $\mathcal{R}$ to be the set of pairs of consecutive terminals on $W$ such that $\gamma$ passes through the contracted vertices $p_{(s,t)}$ exactly for
$(s,t)\in \mathcal{R}$, then after dropping parts of $H'$ not appearing in $H_{\mathcal{R}}$, $\gamma$ becomes a noose enumerated for $H_{\mathcal{R}}$.
Therefore, the terminal partitions raised by $\gamma$ are still included in $\Bb$ as we wanted, and we are done.

\subsection{Traps, issues, and caveats}

The plan sketched in the previous section essentially leads to an algorithm with the promised time complexity, modulo Assumptions A1, A2, A3, and a number of technical details of minor relevance.
Assumptions A2 and A3 are actually quite simple to achieve without loss of generality. It is Assumption A1 that was a major conceptual obstacle.

For Assumption A2, we may at the very beginning perform the following reduction. For every original terminal $t$, introduce a new terminal $t'$ and edges $(t,t')$ and $(t',t)$ of weight $0$ to the graph;
$t'$ and these edges are embedded in any face incident to $t$.
The new terminal set consists of terminals $t'$ for all original terminals $t$.
In this way, any closed walk visiting any new terminal $t'$ has to make a detour of weight $0$ using arcs $(t,t')$ and $(t',t)$, and we may assume that an optimal solution makes only one such detour for
each new terminal $t'$. In this way we can achieve Assumption A2; the actual proof makes a slightly more complicated construction to add a few extra properties.

For Assumption A3, observe that in the reasoning we relied only on the fact that $H_0$ is a tree spanning all terminals that has at most $k$ leaves and at most $k-2$ vertices of degree at least $3$.
In particular, we did not use any metric properties of $H_0$. In fact, the reader may think of $H_0$ as a combinatorial object used to control the homotopy group of the plane with terminals pierced out: for any
non-self-intersecting curve $\gamma$ on the plane, we may infer how terminals are partitioned into those enclosed by $\gamma$, excluded by $\gamma$, and lying on $\gamma$  just by examining the consecutive intersections
of $\gamma$ with $H_0$. Therefore, instead of choosing $H_0$ arbitrarily, we may add it to the graph artificially at the very beginning, say using edges of weight $+\infty$. In this way we make sure that the
optimum solution $W$ does not use any edge of $H_0$.

\begin{figure}[tb]
\begin{center}
\def\svgwidth{\textwidth}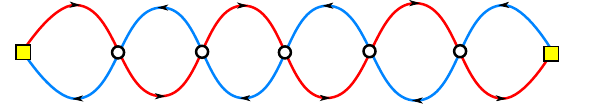
\end{center}
\caption{A planar \tspname{} instance with two terminals. The only solution consists of the union of the red path from $s$ to $t$ and the blue path from $t$ to $s$.
These two paths cross each other many times, which gives many self-intersections of the solution.}\label{fig:oscillator}
\end{figure}

Finally, let us examine Assumption A1: the optimum solution $W$ is a simple directed cycle without self-intersections. Unfortunately, this assumption may not hold in general. 
Consider the example depicted in Figure~\ref{fig:oscillator},
where we have a directed planar graph with two terminals $s,t$, and the only closed walk visiting both $s$ and $t$ consists of two paths, one from $s$ to $t$ and the second from $t$ to $s$, that
intersect each other an unbounded number of times. Therefore, in general the optimum solution $W$ may have an unbounded number of self-intersections. Nevertheless, we may still develop some kind of a combinatorial
understanding of the topology of $W$.

It will be convenient to assume that no edge of the graph is traversed by $W$ more than once; this can be easily achieved by copying each edge $|T|$ times, and using a different copy for each traversal.
Consider two visits of the same vertex $u$ by $W$; let $e_1,e_2$ be the edges incident to $u$ used by $W$ just before and just after the first visit, and define $f_1,f_2$ in the same way for the second visit.
Examine how $e_1,e_2,f_1,f_2$ are arranged in the cyclic order of edges around vertex $u$. If they appear in the interlacing order, i.e., $(e_1,f_1,e_2,f_2)$ or $(e_1,f_2,e_2,f_1)$, then we say that these two
visits form a {\em{self-crossing}} of $W$. Intuitively, if the order is not interlacing, then we may slightly pull the two parts of the embedding of $W$ near $u$ corresponding to the visits so
that they do not intersect. So topologically we do not consider such a self-intersection as a self-crossing. For two walks $W_1,W_2$ in $G$ that do not share any common edges we define their {\em{crossing}} 
in a similar manner, as a common visit of a vertex $u$ such that the cyclic order of edges used by $W_1$ and $W_2$ immediately before and immediately after these visits is interlacing.

We will use the following structural statement about self-crossings of $W$:
We may always choose an optimal solution $W$ so that the following holds. %
\begin{quote}
Consider any self-crossing of $W$ at some vertex $u$ (recall it consists of two visits of $u$) and say it divides $W$ into two closed subwalks $W_1$ and $W_2$: $W_1$ is from the first visit of $u$ to the second,
and $W_2$ is from the second visit of $u$ to the first. Then the subwalks $W_1$ and $W_2$ do not cross at all.
\end{quote}
This statement can be proved by iteratively ``uncrossing'' an optimum solution $W$ as long as the structure of its self-crossings is too complicated.
However, one needs to be careful in order not to split $W$ into two closed curves when uncrossing.

It is not hard to observe that the statement given in the previous paragraph actually shows that the topology of $W$ roughly resembles a cactus where each 2-connected component is a cycle (here, we assume that
self-intersections that are not self-crossings are pulled slightly apart so that $W$ does not touch itself there). See the left panel of Figure~\ref{fig:partition} in Section~\ref{sec:cleaning} for reference.
Then we show (see Lemma~\ref{lem:clean-decomposed}) that $W$ can be decomposed into $\Oh(k)$ subpaths $\Pp=\{B_1,\ldots,B_\ell\}$ such that:
\begin{itemize}
\item each path $B_i$ has no terminal as an internal vertex and is the shortest path between its endpoints; and
\item each path $B_i$ may cross with at most one other path $B_j$.
\end{itemize}
To see this, note that the cactus structure of $W$ may be described as a tree $\tree$ with at most $k$ leaves and at most $k-2$ vertices of degree at least $3$.
We have a pair of possibly crossing subpaths in the decomposition $\Pp$ per each maximal path with internal vertices of degree $2$ in $\tree$. 

The idea now is as follows. 
In the previous section we essentially worked with the partition of $W$ into subpaths between consecutive terminals, as Assumption A1 allowed us to do so.
In the absence of this assumption, we work with the finer partition $\Pp$ as above. The fact that the paths of $\Pp$ interact with each other only in pairs, and in a controlled manner, makes the whole 
reasoning go through with the conceptual content essentially unchanged, but with a lot more technical details.

In the previous description, by Assumption A1, the paths between consecutive terminals do not intersect, and hence they do not interfere with each other while contracting them to three-vertex-paths.
While now the paths $B_i$s may cross, they cross in a very limited setting as described above, causing little turbulence to the argument. 

Another nontrivial difference is that in the previous section we were contracting shortest paths between pairs of consecutive terminals, so we had a small set of candidates for the endpoints of these paths: the terminals
themselves. In the general setting, the decomposition statement above a priori does not give us any small set of candidates for endpoints of paths $B_i$. If we chose those endpoints
as arbitrary vertices of the graph, we would end up with time complexity $n^{\Oh(\sqrt{k})}$ instead of promised $k^{\Oh(\sqrt{k})}\cdot \textrm{poly}(n)$. Fortunately, the way we define 
the decomposition $\Pp=\{B_1,\ldots,B_\ell\}$ allows us to construct alongside also a set of at most $k^4$ {\em{important}} vertices such that each path $B_i$ is the shortest path from one important vertex to
another important vertex.

Finally, there are more technical problems regarding handling possible self-inter-sections of $W$ that are not self-crossings. Recall that in our topological view of $W$, we would like not to regard 
such self-intersections as places where $W$ touches itself. In particular, when examining a sphere-cut decomposition of the union of $W$ and $H_0$ after appropriate contractions, the nooses in this sphere-cut
decomposition should not see such self-intersections as vertices through which they may or should travel. A resolution to this problem is to consider a ``blow-up'' of the original graph where each vertex
is replaced by a large grid and each edge is replaced by a large  matching of parallel edges leading from one grid to another. 
Walks in the original graph naturally map to walks in the blow-up.
Every original self-crossing maps to a self-crossing, and every original self-intersection that is not a self-crossing actually is ``pulled apart'': there is no self-intersection at 
this place anymore. This blow-up has to be performed quite early in the proof.
Unfortunately, while this step is intuitively easy, it does not work very well together with the other simplification steps described above. In particular, it ruins the property of unique shortest paths. 
 Luckily, we are able to extract the essential properties of the blow-up under an abstract definition
of a \emph{canonical instance} and work mostly only with this abstraction.
We will first present a delicate (but self-contained) way of reducing the instances to this form and then we need to solve the problem in simpler, cleaner form.

\begin{figure}[p!]

\centering

\begin{subfigure}{.45\textwidth}
  \begin{center}
{\footnotesize 
\def\svgwidth{\textwidth}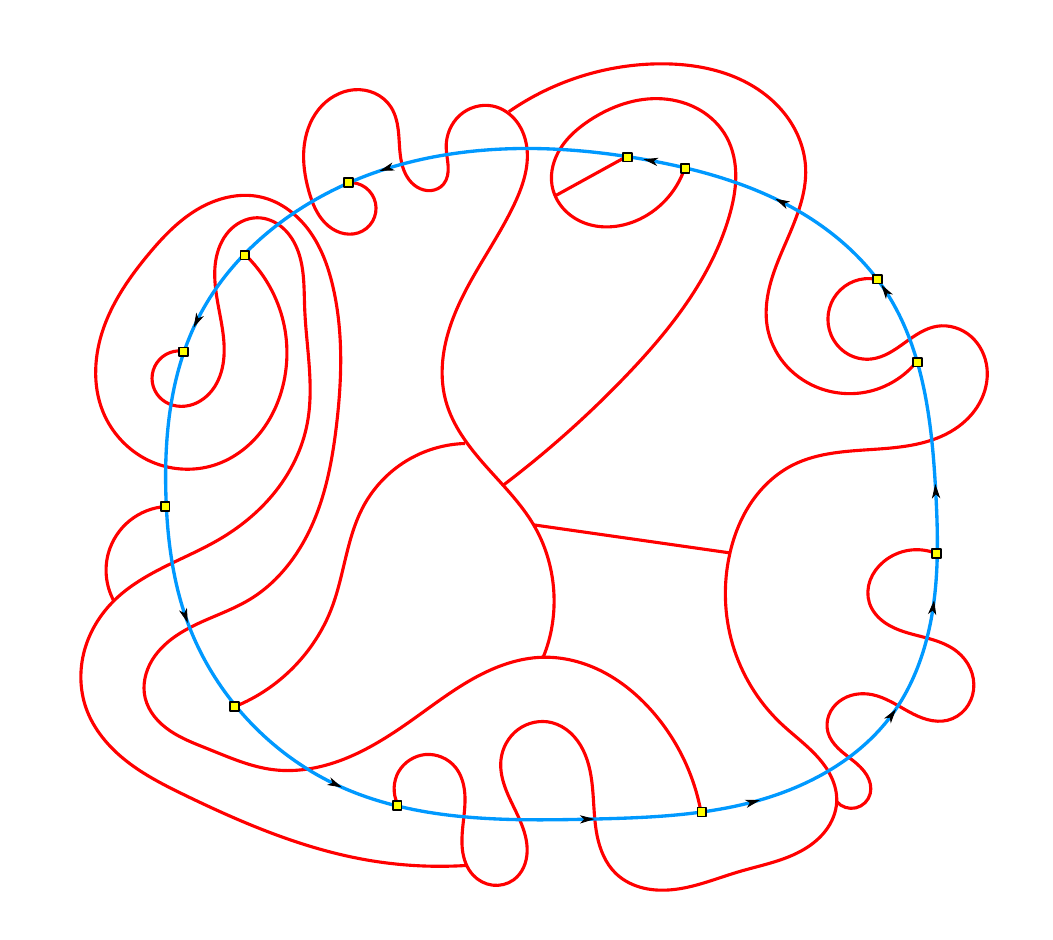
\subcaption{Graph $H$.}\label{p:H}
}
\end{center}
\quad
\end{subfigure}
\begin{subfigure}{.45\textwidth}
  \begin{center}
{\footnotesize 
\def\svgwidth{\textwidth}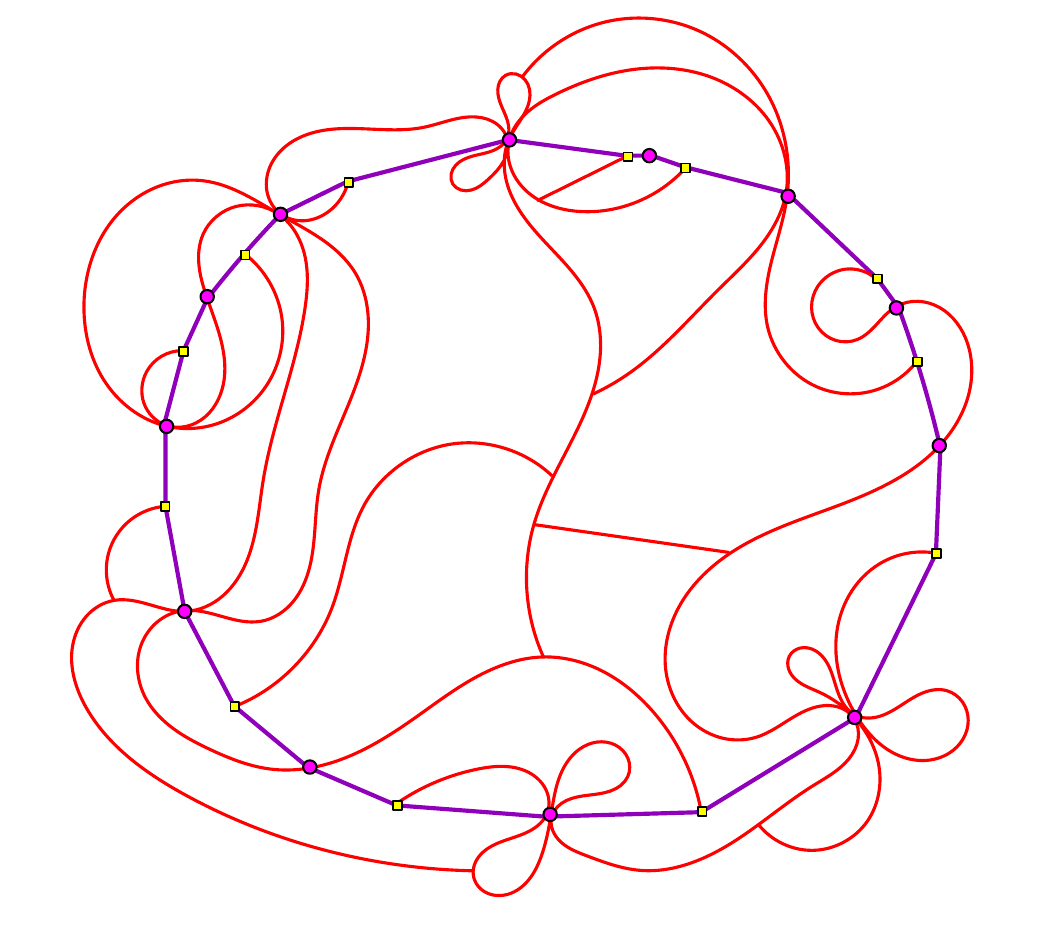
\subcaption{Graph $H'$.}\label{p:Hpp}
}
\end{center}
\end{subfigure}
\begin{subfigure}{.45\textwidth}
  \begin{center}
{\footnotesize 
\def\svgwidth{\textwidth}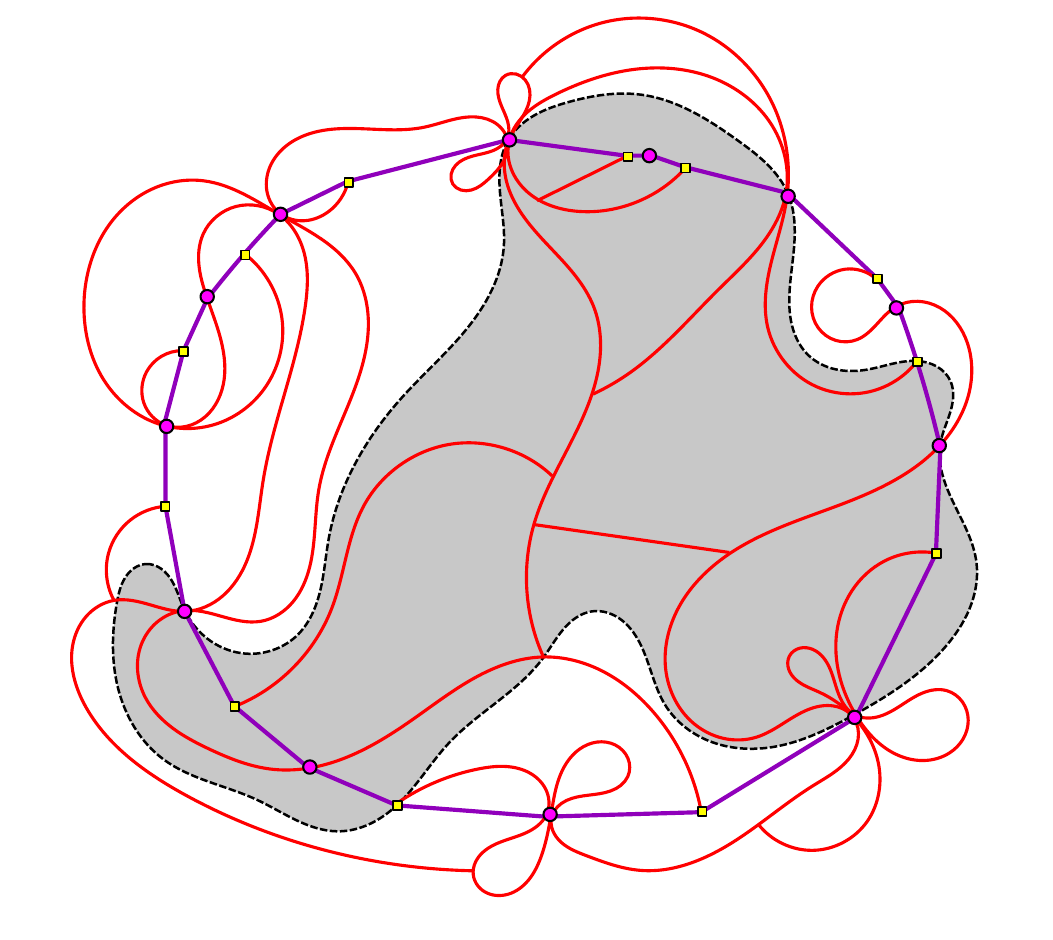
\subcaption{Graph $H'$ with noose $\gamma$.}\label{p:noose}
}
\end{center}
\end{subfigure}
\begin{subfigure}{.45\textwidth}
  \begin{center}
{\footnotesize 
\def\svgwidth{\textwidth}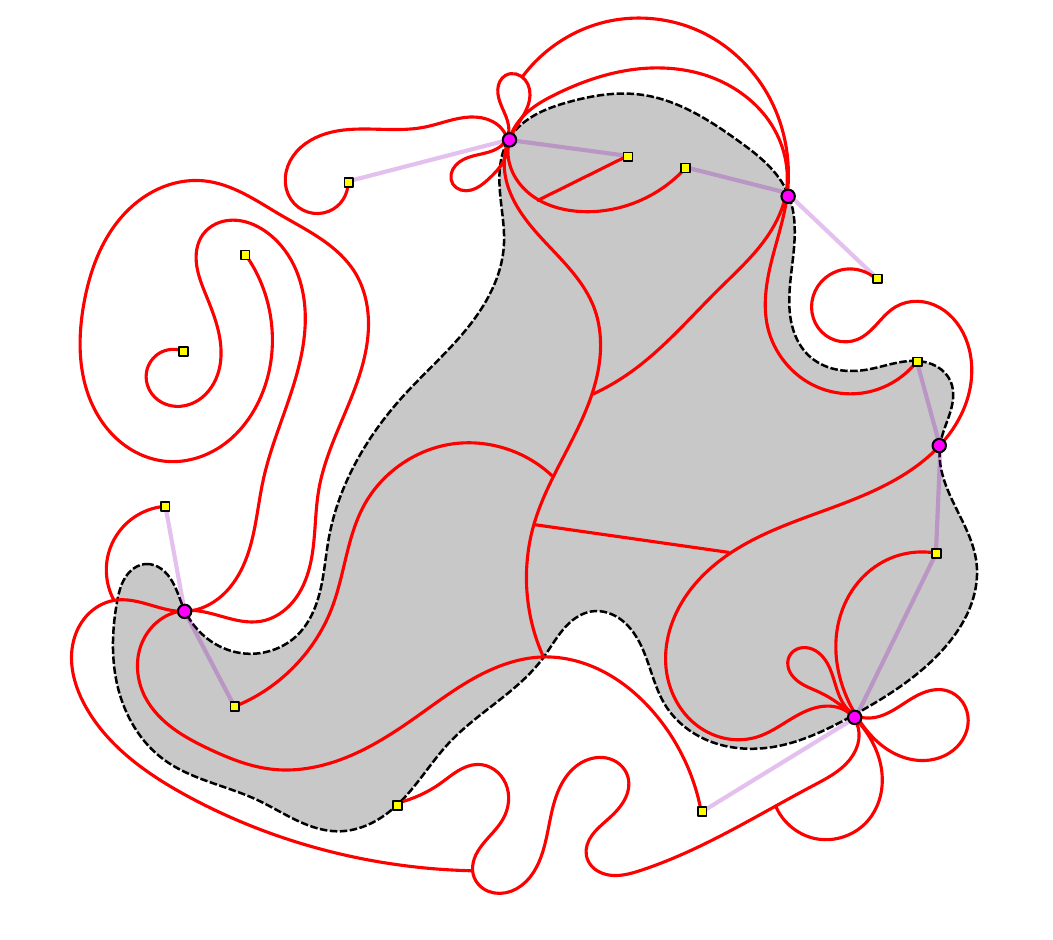
\subcaption{Graph $H_{\mathcal{R}}$ and $\gamma$ as a noose in it.}\label{p:HR}
}
\end{center}
\end{subfigure}
\caption{Construction of Section~\ref{sec:enumeration}. In panel~\ref{p:H}, we see graph $H$ consisting of the union of solution $W$ (in blue) and tree $H_0$ (in red). Terminals are depicted as yellow squares.
In panel~\ref{p:Hpp}, we see graph $H'$, obtained from $H$ by contracting the interior of every subpath of $W$ between two consecutive terminals to one vertex (in violet). Also, paths of vertices of degree $2$ in $H_0$ are replaced by single
edges, though this is not visible. Panel~\ref{p:noose} depcicts noose $\gamma$ in the graph $H'$. Then $\gamma$ partitions the plane into two regions: $R_1$ (grayed) and $R_2$ (non-grayed), which induces a partition of the terminals
into those contained in $R_1$, those contained in $R_2$, and those traversed by $\gamma$. Note that $\gamma$ traverses two terminals, five vertices obtained from contracting subpaths of $W$ to single vertices,
and two vertices of $H_0$ of degree $3$. Finally, in panel~\ref{p:HR} we see the graph $H_{\mathcal{R}}$ used to enumerate $\gamma$. Here, $\mathcal{R}$ consists of those pairs of terminals that
are consecutive on $W$ and moreover $\gamma$ traverses the vertex of $H'$ obtained from contracting the shortest path between them (which is a subpath of $W$). 
These constracted shortest paths are not a part of $H_{\mathcal{R}}$, so they are depicted with reduced opacity.
}\label{fig:overview}
\end{figure}

\section{Preliminaries}\label{sec:prelims}
Throughout the paper  we denote $[n]=\{1,2,\ldots,n\}$  for any positive integer $n$,

We will consider directed or undirected planar graphs $G$ with a terminal set $T \subseteq V(G)$ ($|T|\geq 2$) and weight function $\wei_G\colon E(G) \to \mathbb{Z}_{\geq 0}$;
we omit the subscript if it is clear from the context.
Furthermore, we assume that $G$ does not contain loops, but may contain multiple edges or arcs with the same endpoints.

For a directed path $P$ in a directed graph $G$ and two vertices $u,v \in V(P)$ such that $u$ appears on $P$
not later than $v$, by $P[u,v]$ we denote the subpath of $P$ from $u$ to $v$.
A \emph{$\ell \times \ell$ acyclic grid} consists of vertices $v_{i,j}$ for $1 \leq i,j \leq \ell$, arcs
$(v_{i,j}, v_{i+1,j})$ for every $1 \leq i < \ell$, $1 \leq j \leq \ell$, and arcs
$(v_{i,j}, v_{i,j+1})$ for every $1 \leq i \leq \ell$, $1 \leq j < \ell$.

A \emph{walk} in a directed graph $G$ is a sequence $(e_1,\ldots,e_p)$ of edges of $G$ such that the head of $e_i$ is the tail of $e_{i+1}$, for all $i=1,\ldots,p-1$.
A \emph{walk} is closed if additionally the head of $e_p$ is equal to the tail of $e_1$. The weight of a walk is the sum of the weights of its edges.

\paragraph{Nooses and branch decompositions.}
Given a plane graph $G$, a \emph{noose} is a closed curve without self-intersections that 
meets the drawing of $G$ only in vertices. 
Contrary to some other sources in the literature, we explicitly allow a noose to visit one face multiple times, however each vertex is visited at most once.

We now briefly recall the formal layer of branch and sphere-cut decompositions for convenience.
A \emph{branch decomposition} of a graph $G$ is a pair $(\tree, \brleaf)$
where $\tree$ is an unrooted ternary tree and $\brleaf$ is a bijection between
the leaves of $\tree$ and the edges of $G$. For every edge $e \in E(\tree)$,
 we define the \emph{cut} (or \emph{middle set}) $\brcut(e) \subseteq V(G)$ as follows:
if $\tree_1$ and $\tree_2$ are the two components of $\tree-e$, then $v \in \brcut(e)$
if $v$ is incident both to an edge corresponding to a leaf in $\tree_1$ and to an edge corresponding
to a leaf in $\tree_2$.
The width of a decomposition is the maximum size of a cut in it, and the branchwidth
of a graph is a minimum width of a branch decomposition of a graph.
It is well known that planar graphs have sublinear branchwidth.
\begin{theorem}[see e.g. \cite{FominT06}]\label{thm:planar-branch}
Every planar graph with $n \geq 2$ vertices of degree at least $3$ has branchwidth bounded by $\sqrt{4.5 n}$.
\end{theorem}

In planar graphs, one can compute good branch decompositions, where the cuts $\brcut(e)$ correspond
to nooses. More formally, a triple $(\tree, \brleaf, \brnoose)$ is an \emph{sc-branch decomposition}
(for \emph{sphere-cut branch decomposition}) if $(\tree, \brleaf)$ is a branch decomposition
and for every $e \in E(\tree)$, $\brnoose(e)$ is a noose that traverses the vertices of $\brcut(e)$
and separates the edges corresponding to the leaves of the two components of $\tree-e$ from
each other.

We need the following result of Seymour and Thomas~\cite{SeymourT94}, with the algorithmic part following from~\cite{GuT08,DornPBF10}. We remark that the $\Oh(|V(G)|^3)$ factor coming from this
theorem is a part of a the polynomial factor in the running time bound of our algorithm (that we do nnot analyse in detail). 
\begin{theorem}[\cite{SeymourT94,GuT08,DornPBF10}]\label{thm:scbranch}
Given a connected plane graph $G$, %
one can in time $\Oh(|V(G)|^3)$ compute an sc-branch decomposition of $G$
of width equal to the branchwidth of~$G$.
\end{theorem}

We remark that in \cite{SeymourT94,GuT08,DornPBF10} one considers nooses that can visit every face at most once, which makes it necessary to assume also that the graph is bridgeless; see e.g.~\cite{MarxP15}.
It is easy to see that without this assumption on nooses, one can extend the theorem also to connected graphs with bridges.
One way to obtain it is to first decompose into bridgeless components, and then decompose each such component separately.
Alternatively, one can add a number of dummy edges without violating the plane embedding but ensuring $2$-edge-connectivity.

\section{Nooses}\label{sec:nooses}
In this section, we prove a combinatorial result showing that if we consider nooses that go through only a limited number of vertices of a connected graph with some vertices being terminals,
then there is only a bounded number of potential partitions of terminals such nooses can realize. A slight technical complication is deciding how to handle terminals that are on the noose itself; to avoid this complication, we consider the terminals to be edges instead.

Let $G$ be a connected plane (directed or undirected) graph 
with a set $\Tedges \subseteq E(G)$ of \emph{terminal edges}
and let $\gamma$ be a noose in $G$ that visits at most $\ell$ vertices. 
In this section we show that if $\ell \ll |\Tedges|$, then there 
are much less than $2^{\Theta(|\Tedges|)}$ ways of how the noose can partition the set of terminal edges.

More formally, we think of the planar embedding of $G$ as a spherical one
(i.e., without distinguished outer face)
  and with a noose $\gamma$ we associate a \emph{partition} 
  $\{\Tedges_1,\Tedges_2\}$ of $\Tedges$, 
where and $\Tedges_1$ and $\Tedges_2$ are the sets of terminal edges that lie in the two components
of the sphere minus $\gamma$. 
Since we consider spherical embeddings and the two sides of $\gamma$ are symmetric,
the pair $\{\Tedges_1,\Tedges_2\}$ is an unordered pair.

Our main claim in this section is that there are only $|\Tedges|^{\Oh(\ell)}$ ``reasonable''
partitions for nooses visiting at most $\ell$ vertices.

\begin{lemma}\label{lem:noose-enum}
Assume we are given a plane connected graph $G$ with a set $\Tedges \subseteq E(G)$ of terminal edges and an integer~$\ell$. 
Then one can in time $|\Tedges|^{\Oh(\ell)} n^{\Oh(1)}$ compute a family $\Aa$ of
$|\Tedges|^{\Oh(\ell)}$ of partitions of $\Tedges$ such that, for every noose of $G$ that visits at most
$\ell$ vertices, its corresponding partition of the terminal edges belongs to~$\Aa$.
\end{lemma}
\begin{proof}

The crucial observation is that deleting an edge or a (nonterminal) vertex from $G$ only increases
the family of curves in the plane that are nooses with respect to $G$. Consequently, if one replaces $G$ with any of its
connected subgraphs that contains all the terminal edges and enumerate a family of partitions satisfying the statement of the lemma for this subgraph,
then the same family will be also a valid output for the original graph $G$.
Thus, by restricting attention to an inclusion-wise minimal connected subgraph containing all terminal edges, without loss of generality we may assume that every edge of $G$ that is not a terminal edge
is a bridge connecting two parts of $G$ that both contain a terminal edge.

Without loss of generality, assume $\ell < |\Tedges|$, as otherwise we just enumerate all partitions of $\Tedges$.

Let $S$ be the set of \emph{special} vertices in $G$: endpoints of terminal edges
and vertices of degree at least $3$. Note that every vertex of $S \setminus V(\Tedges)$
is a vertex incident with at least three nonterminal edges; each such edge is a bridge
connecting two components containing a terminal edge. Hence, $|S \setminus V(\Tedges)| < |\Tedges|$
and thus $|S| < 3|\Tedges|$. 
Furthermore, it follows that $G$
decomposes into $\Tedges$ and $r < 2|\Tedges|$ paths $Q_1,Q_2,\ldots,Q_r$
such that each path $Q_i$ consists of nonterminal edges only, 
has both endpoints in $S$ but no internal vertices in $S$.
That is, every path $Q_i$ is disjoint from $\Tedges$, has degree-2 vertices as internal vertices,
and either endpoints of terminal edges or vertices of degree at least $3$ as endpoints.

Construct now a graph $G'$ from $G$ by replacing every path $Q_i$ with a path $Q_i'$ with the same drawing in the plane,
but with exactly $\ell$ internal vertices.
We have 
\begin{align*}
|V(G')| &\leq |S| + \ell \cdot r < |\Tedges|(2\ell + 3),\\
|E(G')| &\leq |\Tedges| + (\ell+1) \cdot r < |\Tedges|(2\ell + 3).
\end{align*}
Furthermore, for every noose $\gamma$ in $G$ that visits at most $\ell$ vertices of $G$, construct its \emph{shift} $\gamma'$,
being a noose with respect to $G'$, as follows: for every path $Q_i$, move all intersections of $\gamma$ with the internal vertices
of $Q_i$ to distinct internal vertices of $Q_i'$, keeping the relative order of the intersections along the paths $Q_i$ and $Q_i'$ the same.
Since $Q_i'$ has $\ell$ internal vertices, this is always possible. 
Furthermore, we can obtain $\gamma'$ from $\gamma$ by local modifications
within close neighborhoods of the paths $Q_i$, but not near its endpoints. Consequently, the partitions of the terminal edges induced
by $\gamma$ and $\gamma'$ are the same.

Observe now that $\gamma'$ is a noose with respect to a connected graph with $\Oh(|\Tedges| \ell)$ vertices and edges.
With every intersection of $\gamma'$ with $G'$, say at a vertex $v$, we associate three pieces of information: the vertex $v$ itself,
between which pair of edges incident with $v$ the noose $\gamma'$ entered $v$, and between which pair of edges it left $v$.
Since there are only $\Oh(|\Tedges|\ell) = \Oh(|\Tedges|^2)$ choices for every piece of information,
there are only $|\Tedges|^{\Oh(\ell)}$ possible \emph{combinatorial representations} of $\gamma'$, defined as a sequence
of the aforementioned triples of pieces of information at every vertex traversed $\gamma'$, in the order of a walk along~$\gamma'$.
Finally, as the connectedness of $G'$ implies that every face of $G'$ is isomorphic to a disc, we can see that knowing the combinatorial representation of $\gamma'$
is sufficient to deduce the partition of the terminal edges induced by~$\gamma'$.
This finishes the proof.
\end{proof}

\section{The algorithm}\label{sec:tsp}

In this section we provide a full proof of Theorem~\ref{thm:dirtsp}. 
We assume that we are given an instance $(G,T)$ of \tspname{}.
We start by fixing a plane embedding of $G$ and introducing a few useful definitions.

Let $W$ be a walk that visits every terminal exactly once.
A permutation $\pi = (t_1,t_2,\ldots,t_{|T|})$ of $T$ is a \emph{witnessing permutation}
of $W$ if it is exactly the (cyclic) order of the terminals visited by $W$.
A closed walk $W$ is a \emph{locally short walk} if it visits every terminal exactly once
and the subwalks of $W$ between the consecutive terminals are actually shortest paths between their endpoints.

For two edge-disjoint paths $P_1$, $P_2$ and a nonterminal vertex $v \in V(P_1) \cap V(P_2)$
we say that $v$ is a \emph{transversal intersection} of $P_1$ and $P_2$ if $v$ is not an endpoint
of neither $P_1$ nor $P_2$ and if $e_i^1$, $e_i^2$ are the two edges of $P_i$ incident with $v$
for $i=1,2$, then they are in the order $e_1^1,e_2^1,e_1^2,e_2^2$ clockwise or counter-clockwise
around $v$.

We proceed in a number of steps.
The crucial definition that allows us to control self-crossings of the solution via a ``cactus-like'' structure is the following.

\begin{definition}\label{def:reduced}
Suppose $W$ is a closed walk in $G$ that visits every terminal at most once.
Then $W$ is called {\em{cactuslike}} if every terminal is visited by $W$ exactly once, 
     every vertex of $G$ is visited by $W$ at most twice, and, moreover, the following condition holds.
Whenever a vertex $x$ is visited twice by $W$, then the two proper subwalks of $W$ obtained by following $W$ from one visit of $x$ to the other have no intersection other than $x$.
\end{definition}

\begin{figure}[tb]
\begin{center}
\def\svgwidth{0.4\textwidth}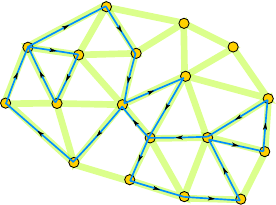
\end{center}
\caption{An example cactuslike walk (in blue). There are two vertices at which the walk self-crosses and two vertices that
are visited twice without self-crossing.}\label{fig:cactus}
\end{figure}

An example cactuslike walk is depicted in Figure~\ref{fig:cactus}. As the reader may see, the walk has a shape roughly resembling a cactus, or more formally a tree consisting of pairs of interlacing directed paths. 

In Section~\ref{sec:bunch} we study the notion of a \emph{canonical instance} $(\Ggrid,T,\Pp)$ %
of \tspname{}.
The main of this notion is to exclude some degenerate scenarios, such as the solution visiting a vertex more than twice, using the same edge twice, or intersecting itself without a good reason. 
However, to exclude the above degenerate scenarios, we cannot at the same time ensure the shortest path property of the instance; instead, we offer a family of \emph{canonical paths} between 
terminals that need to be used by the solution.
\begin{definition}\label{def:blow-up}
A triple $(\Ggrid, T, \mathcal{P})$ is a \emph{canonical instance} of \tspname{}
if $(\Ggrid,T)$ is a \tspname{} instance and 
$\mathcal{P}$ is a family of $|T|(|T|-1)|$ paths $\Pgrid(s,t)$ for every terminal pair $(s,t)$, $s \neq t$ 
(called the \emph{canonical $s-t$ path}) with the following properties:
\begin{enumerate}[(A)]
\item the path $\Pgrid(s,t)$ does not visit any other terminals and is a shortest
path from $s$ to $t$;\label{i:blow-up:path}
\item the paths $\Pgrid(s,t)$ are pairwise edge-disjoint and every nonterminal vertex of $\Ggrid$ lies on at most two paths $\Pgrid(s,t)$;\label{i:blow-up:disjoint}
\item if two paths $\Pgrid(s_1,t_1)$ and $\Pgrid(s_2,t_2)$ intersect at a nonterminal vertex $x$,
  then $s_1 \neq t_1$, $s_2 \neq t_2$, and the intersection at $x$ is transversal;\label{i:blow-up:intersection}
\item there exists a minimum-weight solution $\Wgrid$ to \tspname{} on $(\Ggrid,T)$
that is a concatenation of $|T|$ canonical paths (and thus is locally short) and that is cactuslike.\label{i:blow-up:canonical}
\end{enumerate}
\end{definition}
A solution to \tspname{} on $(\Ggrid,T,\mathcal{P})$ is \emph{canonical} if it is a concatenation of $|T|$ canonical paths. 
Note that a canonical solution visits every terminal exactly once and every nonterminal vertex at most twice (in particular, it is locally short).

\medskip

In Section~\ref{sec:bunch}, we describe how to turn the input instance $(G,T)$
into a canonical instance $(\Ggrid,T,\mathcal{P})$ with the existence of the canonical cactuslike solution
$W$ 
proven in Section~\ref{sec:cleaning}. That is, we prove the following statement.
\begin{lemma}\label{lem:input-to-blow-up}
Given a \tspname{} instance $(G,T)$, one can in polynomial time
compute a canonical instance 
$(\Ggrid,T,\mathcal{P})$ with $|E(\Ggrid)| + |V(\Ggrid)| \leq (|E(G)| + |V(G)|)|T|^{\Oh(1)}$,
        with the same terminal set $T$ and with the following properties:
\begin{itemize}
\item given a solution $\Wgrid$ that is a solution to \tspname{}
in $(\Ggrid,T,$ $\mathcal{P})$ of minimum possible weight, one can in polynomial time find
a solution $W$ to \tspname{} in $(G,T)$ that is of minimum possible weight;
\item given a solution $W$ to \tspname{} in $(G,T)$ that is of minimum possible weight
one can in polynomial time find
a solution $\Wgrid$ to \tspname{} in $(\Ggrid,T,\mathcal{P})$ that is of minimum possible weight.
\end{itemize}
\end{lemma}
Note that Lemma~\ref{lem:input-to-blow-up} reduces the \tspname{} problem on $(G,T)$
to \tspname{} on a canonical instance $(\Ggrid,T,\mathcal{P})$. Thus it is sufficient to solve the canonical version of problem.

In Section~\ref{sec:reduced} we formalize the intuition that the notion of a \emph{cactuslike} walk
gives a cactus-like structure on the walk. Sections~\ref{sec:enum} and~\ref{sec:dp} give an algorithm for \tspname{} on canonical instances.
\begin{lemma}\label{lem:alg-blow-up}
Given a canonical instance $(\Ggrid,T,\mathcal{P})$, one can in time \[ |T|^{\Oh(\sqrt{|T|})} |\Ggrid|^{\Oh(1)} \]
find a minimum weight solution $\Wgrid$ to \tspname{} on $(\Ggrid,T,\mathcal{P})$. 
\end{lemma}
In Section~\ref{sec:enum} we show how to enumerate a small family of states for a dynamic 
programming algorithm and then in Section~\ref{sec:dp} we present the dynamic programming routine
itself. 
By pipelining Lemmas~\ref{lem:input-to-blow-up} and~\ref{lem:alg-blow-up} one derives
Theorem~\ref{thm:dirtsp}.

\subsection{Constructing a canonical instance}\label{sec:bunch}
\newcommand{\cross}{\mathrm{cross}}

\paragraph{Initial preprocessing}
We start with the following preprocessing steps on $G$.

First, we ensure that shortest paths in the input instance $G$ are unique
and that the edge weights are strictly positive.
Since we do not analyze the polynomial factor in the running time bound of our algorithms, this can be ensured in a standard manner
by replacing a weight $\wei(e)$ of the $i$-th arc with
$M \cdot (\wei(e) \cdot n^{|E(G)|+1} + n^i)$ for $M = 2|T|(|E(G)| +1)+1$ 
(the factor $M$ is used to later add some weight-1 edges without changing the structure 
 of the minimum-weight solution).
Let $P_G(u,v)$ be the unique shortest path from $u$ to $v$ in $G$.  

Second, we ensure that every terminal $t \in T$ has only one neighbor $w_t$, with two
arcs $(w_t, t)$ and $(t,w_t)$ of weight $1$.
To obtain such a property, for every terminal $t_0 \in T$ we can make its copy $t$, connect
$t$ and $t_0$ with arcs in both direction of weight $1$, and rename $w_t = t_0$. The new terminal set is the set of the copies of the old terminals.
Note that this property implies that 
we can consider only solutions to the \tspname{}
problem that visit every terminal exactly once. 
Note also that this operation does not spoil the property that $G$ has unique shortest paths.

Third, we ensure that every nonterminal vertex $v$ has in-degree $1$ and out-degree $2$ or in-degree $2$ and out-degree $1$. 
To this end, we first iteratively remove all nonterminal vertices of in- or out-degree $0$; they surely are not used in any solution.
Then, for every remaining nonterminal vertex $v$ with $d$ edges, we replace $v$ with a directed cycle of length $d$ and each arc of weight $1$, attaching every arc incident with $v$
to a different vertex on the cycle. We perform this operation so that the graph remains planar: for the fixed embedding of $G$ we attach arcs incident with $v$
in the cyclic order in this embedding.

Observe that for a terminal $t$ with a sole neighbor $w_t$, after this operation the terminal $t$ is still incident with two arcs, $(w_1,t)$ and $(t,w_2)$ where $w_1$ and $w_2$
are two consecutive vertices on the cycle corresponding to the vertex $w_t$. Furthermore, $w_1$ has out-degree $2$ and in-degree $1$ while $w_2$ has in-degree $2$ and out-degree $1$.
Again, we also observe that this operation does not spoil the property that $G$ has unique shortest paths. Here, the crucial fact is that we put weight $1$ (as opposed to $0$) on the arcs
incident through a terminal, so a detour from $w_1$ to $w_2$ via $t$ is more expensive than following the (weight-$1$) arc $(w_1,w_2)$ directly. See Figure~\ref{fig:tsp:preprocess} for an illustration.

\begin{figure}[tb]
\begin{center}
\includegraphics[width=\linewidth]{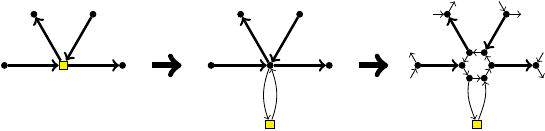}
\end{center}
\caption{Initial preprocessing for \tspname{}. The new edges, marked as thinner lines, are of weight $1$, much smaller than the original (thicker) arcs.}\label{fig:tsp:preprocess}
\end{figure}

Finally, note that after this operation the length of any terminal-to-terminal path in $G$
increased its weight by at most $2|E(G)|+2$ and hence the length of any solution to \tspname{}
increased its weight by at most $2|T|(|E(G)|+1) = M-1$. Thus, a minimum-weight solution
in the modified graph will project back to a minimum-weight solution in the original graph
and vice versa. By somehow abusing the notation, we keep $(G,T)$ as the name for the instance after the above initial 
preprocessing.

To sum up, by the above operations we ensure that in the input graph $G$ is embedded on a plane and:
\begin{enumerate}
\item the shortest paths in $G$ are unique and the edge weights in $G$ are positive integers;
\item every terminal $t$ is of in-degree $1$ and out-degree $1$ with incident edges $(w_1, t)$ and $(t,w_2)$ such that $(w_1,w_2)$ is an arc of weight $1$
and the three arcs $(w_1,t)$, $(t,w_2)$, $(w_1, w_2)$ bound a face;
\item for every two vertices $v, w \in V(G)$, the unique  shortest path between its endpoints does not visit any terminal as an internal vertex;
\item every nonterminal vertex $v$ is of in-degree $2$ and out-degree $1$ or out-degree $2$ and in-degree $1$; we henceforth call the vertices
of the first type the \emph{in-2 out-1 vertices} while of the second type the \emph{in-1 out-2 vertices}. 
\end{enumerate}

\paragraph{Construction of the canonical instance}
We now move to the construction of the canonical instance $(\Ggrid, T,\mathcal{P})$.

Recall that $P_G(u,v)$ denotes the (unique) shortest path from $u$ to $v$ in $G$. 
Let $\mathcal{T} = \{(s,t) \in T \times T~|~s \neq t\}$ and $\mathcal{P} = \{P_G(s,t)~|~(s,t) \in \mathcal{T}\}$. 
For every edge $e \in E(G)$, let 
$\mathcal{T}(e) = \{(s,t) \in \mathcal{T}~|~e \in E(P_G(s,t))\}$
and $\mathcal{P}(e) = \{P \in \mathcal{P}~|~e \in E(P)\} = \{P_G(s,t)~|~(s,t)\in \mathcal{T}(e)\}$.
Similarly define $\mathcal{T}(v)$ and $\mathcal{P}(v)$ for a nonterminal vertex $v$.

Let $e=(u,v) \in E(G)$.
Let $T_e^\rightarrow \subseteq T$ be the family of terminals $t \in T$ 
for which there exists some $(s_1,t_1) \in \mathcal{T}(e)$ with $t_1 = t$.
Similarly, let 
$T_e^\leftarrow \subseteq T$ be the family of terminals $s \in T$ 
for which there exists some $(s_1,t_1) \in \mathcal{T}(e)$ with $s_1 = s$.
Consider the union $H_e^\rightarrow$ of all shortest paths $P_G(v,t)$ for all $t \in T_e^\rightarrow$.
It is clear that $H_e^\rightarrow$ is an outbranching rooted at $v$.
Subdivide for a moment the edge $e$ with a new vertex $x_e$, add the arc $(x_e, v)$ to $H_e^\rightarrow$ and consider $H_e^\rightarrow$ as an outbranching rooted at $x_e$.
Since no shortest path in $G$ contains a terminal as an internal vertex, while if $v \in T$ then $T_e^\rightarrow = \{v\}$, 
all terminals are leaves of $H_e^\rightarrow$.
Consequently, $H_e^\rightarrow$ imposes an order $\preceq_e^\rightarrow$ on $T$: starting from the root, we traverse the unique face of $H_e^\rightarrow$ in counter-clockwise direction and order $T$ in the order of this traversal.
The order $\preceq_e^\rightarrow$ is the \emph{destination order} for the edge $e$. 

Similarly we define the order $\preceq_e^\leftarrow$ by taking $H_e^\leftarrow$ to be an inbranching rooted at $x_e$ consisting of the edge $(u, x_e)$ and all shortest paths $P_G(t, u)$ for $t \in T_e^\leftarrow$ and traversing the unique face of $H_e^\leftarrow$ in the clockwise direction (note that if $u \in T$, then $T_e^\leftarrow = \{u\}$).
The order $\preceq_e^\leftarrow$ is the \emph{source order} for the edge $e$.
Finally, we define $\preceq_e$ as an order on $\mathcal{T}(e)$ where we order all pairs $(s,t) \in \mathcal{T}(e)$ lexicographically first by the destination order $\preceq_e^\rightarrow$ for $e$ of $t$ and then
by the source order $\preceq_e^\leftarrow$ for $e$ of $s$. That is, $(s_1,t_1) \preceq_e (s_2,t_2)$ if and only if $t_1 \prec_e^\rightarrow t_2$ or $t_1=t_2$ and $s_1 \preceq_e^\leftarrow s_2$.

We define the graph $\Gbunch$ as $G$ with every arc $e = (u,v)$ replaced with $|\mathcal{T}(e)|$ parallel copies (of the same weight, drawn next to each other in the plane). Furthermore, we label the copies of $e$ with distinct elements of $\mathcal{T}(e)$: we go around $u$ in the counter-clockwise order and label the copies according to the order $\preceq_e$. (Note that we will obtain the same labelling if we go around $v$ in the clockwise order.) The set of all copies of $e$ in $\Gbunch$ is called the \emph{bunch of $e$}. 
Finally, for every path $P_G(s,t)$ for $(s,t) \in \mathcal{T}$, we define 
\emph{the canonical path in $\Gbunch$}, denoted $\Pbunch(s,t)$, as the path that for every $e \in E(P_G(s,t))$ traverses the copy of $e$
assigned label $(s,t)$. Note that $V(G) = V(\Gbunch)$.

Finally, we define the graph $\Ggrid$ as follows. Start with $\Gbunch$ and for every
nonterminal vertex $v$ proceed as follows.
Assume first that $v$ is in $G$ 
an in-2 out-1 vertex with incident edges $e_1 = (u_1, v)$, $e_2 = (u_2, v)$ and $e = (v, w)$
lying around $v$ in this counter-clockwise order.
Replace $v$ with a $|\mathcal{T}(v)| \times |\mathcal{T}(v)|$ 
acyclic grid $\grid(v)$ with all arcs of weight $0$. 
Assume that $\grid(v)$ is drawn such that all edges go rightwards and upwards.
In $\Ggrid$, we attach the edges of $\Gbunch$ incident with $v$ as follows.
Attach all $|\mathcal{T}(v)| = |\mathcal{T}(e)|$ copies of $e$ to the right side of $\grid(v)$,
one edge per vertex of the grid, in the same order as the cyclic order around $v$ in $\Gbunch$
(that is, in the order $\preceq_e$ from bottom to top).
The label of a row of $\grid(v)$ is the label of the edge outgoing from the right endpoint of the row.
Attach all $|\mathcal{T}(e_1)|$ copies of $e_1$ to the left side of $\grid(v)$
so that a copy with label $(s,t)$ is attached to the vertex in the row with the same label.
Attach all $|\mathcal{T}(e_2)|$ copies of $e_2$ to the bottom side of the grid,
at most one copy per vertex of $\grid(v)$, in the same order as the cyclic order
around $v$ in $\Gbunch$ (that is, in the order $\preceq_{e_1}$ from right to left). 
Again, for each column where a copy of $e_2$ is attached to the bottom vertex of the column,
  the label of the column is the label of the attached copy of $e_2$.
The label of an edge of $\grid(v)$ is the label of its row and column; note that some vertical edges do not receive any label if no edge is attached to the bottom endpoint of the column.

The construction for in-1 out-2 vertex $v$
is symmetric with all directions of arcs reversed (see Figure~\ref{fig:canonical:construct}). 
This finishes the description of the graph $\Ggrid$; note that we do not modify the terminals.
Again, as in $\Gbunch$, the set of all copies of an edge $e \in E(G)$ in $\Ggrid$ is called the \emph{bunch} of $e$.

For every path $P_G(s,t)$ for $(s,t) \in \mathcal{T}$, we define
\emph{the canonical path in $\Ggrid$}, denoted $\Pgrid(s,t)$, as follows.
For every $e \in E(P_G(s,t))$, we traverse the copy of $e$ labeled $(s,t)$.
For every nonterminal vertex $v$ on $P_G(s,t)$ with preceding edge $e = (u,v)$ and
succeeding edge $e' = (v,w)$ on $P_G(s,t)$, we connect the head of the copy of $e$
labeled $(s,t)$ with the tail of $e'$ labeled $(s,t)$ as follows. 
If these copies are attached to opposite sides of $\grid(v)$, we connect via the corresponding
row of $\grid(v)$ that has label $(s,t)$. Otherwise, if these copies are attached to perpendicular sides,
we connect with an L-shape via the row and column labeled $(s,t)$, taking only one turn at the intersection of the row and column
labeled $(s,t)$. See Figure~\ref{fig:canonical:construct} for an illustration.

\begin{figure}[t]
\begin{center}
\includegraphics[width=0.9\linewidth]{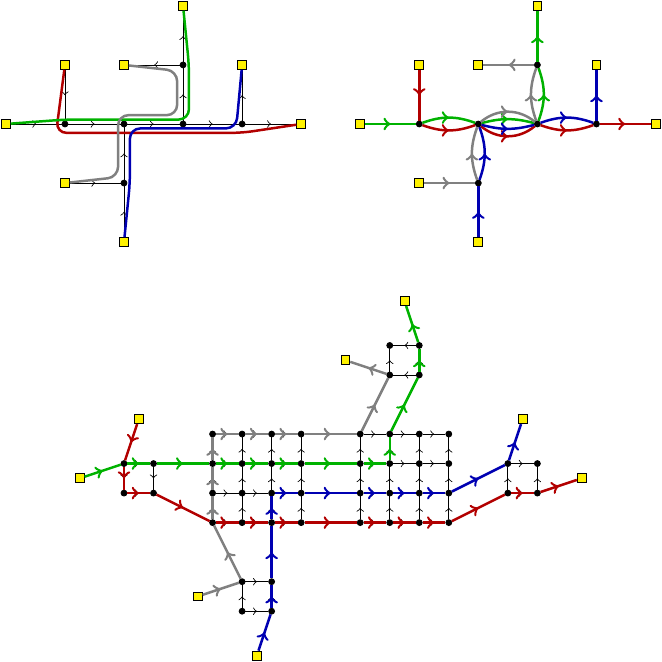}
\end{center}
\caption{The construction of $\Gbunch$ and $\Ggrid$. 
 Top left: a subgraph of the graph $G$ with four shortest terminal to terminal paths and every nonterminal vertex being either an in-2 out-1 vertex or an in-1 out-2 vertex.
 Top right: the corresponding subgraph of $\Gbunch$ with the lifts of the four paths. (The formal definition of a \emph{lift} appears later in text.)
 Bottom: the corresponding subgraph of the graph $\Ggrid$ with the lifts of the four paths.}\label{fig:canonical:construct}
\end{figure}

\paragraph{Basic properties}
It is straightforward to observe that in this manner every arc $e$ of $\Ggrid$ is used in at most
one canonical path $\Pgrid(s,t)$ and, if this is the case, then $e$ is assigned label $(s,t)$.
From the fact that no shortest path of $G$ contains
a terminal as an internal vertex we infer that canonical paths in neither $\Gbunch$ nor
$\Ggrid$ contain a terminal as an internal vertex. Since all nonterminal vertices of $\Ggrid$
have in- and out-degree bounded by $2$, and, for vertices with both in- and out-degree equal $2$, 
     the labels of the incident arcs alternate, 
we infer that at every nonterminal vertex of $\Ggrid$ at most two canonical paths can intersect
and, if they intersect, they intersect transversally. 
This proves properties~\ref{i:blow-up:disjoint}, \ref{i:blow-up:intersection},
and the first part of property~\ref{i:blow-up:path}
of Definition~\ref{def:blow-up}.

Note that in $\Gbunch$ and $\Ggrid$ we no longer have unique shortest paths, but 
we do not need them in these graphs;
in what follows we will rely on the notion of canonical paths instead.

Observe that every walk in $\Gbunch$ or $\Ggrid$ has its natural \emph{projection} in $G$ of the
same weight. This proves the second part of property~\ref{i:blow-up:path} of Definition~\ref{def:blow-up}. 
Furthermore, it is easy to observe that if a projection of a walk in $\Gbunch$ or in $\Ggrid$
is a simple path in $G$, then the preimage walk needs to be a simple path as well.

Consider a locally short walk $W$ in $G$ and let $(t_1,t_2, \ldots, t_{|T|})$ be the witnessing
permutation. That is, $W$ is the concatenation of paths $P_G(t_i,t_{i+1})$ for $i=1,2,\ldots,|T|$
(where $t_1 = t_{|T|+1}$). 
The notion of canonical paths allows us to lift $W$ to $\Wbunch$ in $\Gbunch$
and $\Wgrid$ in $\Ggrid$: $\Wbunch$ is a concatenation of
paths $\Pbunch(t_i,t_{i+1})$ while $\Wgrid$ is a concatenation of
paths $\Pgrid(t_i,t_{i+1})$. 
Note that both $\Wbunch$ and $\Wgrid$ use every edge of the corresponding
graph at most once and that the weights of $W$, $\Wbunch$, and $\Wgrid$ are equal.

In the other direction, if we have a solution $\Wgrid$ to \tspname{} in $(\Ggrid,T,\mathcal{P})$, then
there is a natural way to project $\Wgrid$ back to a solution $W$ to $(G,T)$:
whenever $\Wgrid$ traverses a terminal $t$ or a grid $\grid(v)$, go through $t$
or $v$ in $G$, respectively. Note that the weight of $W$ is not larger
than the weight of $\Wgrid$.
This shows the first part of property~\ref{i:blow-up:canonical} of Definition~\ref{def:blow-up}
(i.e., except for the ``cactuslike'' claim). 

Thus, to prove Lemma~\ref{lem:input-to-blow-up}, it remains to show the existence
of a canonical cactuslike solution of minimum weight to \tspname{} on $(\Ggrid,T,\mathcal{P})$; all other required properties
of the canonical instance and properties promised by Lemma~\ref{lem:input-to-blow-up} have been discussed
above or are straightforward.

To this end,
   we observe that two paths $\Pgrid(s_1,t_1)$ and $\Pgrid(s_2,t_2)$ intersect 
only in a very specific situation. 
Intuitively, two paths $P_G(s_1,t_1)$ and $P_G(s_2,t_2)$ can intersect and share edges multiple times;
for each such common subpath of $P_G(s_1,t_1)$ and $P_G(s_2,t_2)$, the paths
$\Pgrid(s_1,t_1)$ and $\Pgrid(s_2,t_2)$ follow the same grids $\grid(\cdot)$ and bunches in parallel,
crossing only at the first grid and only if the corresponding intersection of $P_G(s_1,t_1)$ and $P_G(s_2,t_2)$ is transversal after
contracting the edges of the common subpath.

\begin{lemma}\label{lem:grid-cross}
Let $(s_1,t_1)$ and $(s_2,t_2)$ be two distinct elements of $\mathcal{T}$ and let $v$
be a nonterminal vertex of $G$.
Then $\Pgrid(s_1,t_1)$ and $\Pgrid(s_2,t_2)$ intersect at at most one vertex of $\grid(v)$.
Furthermore, such an intersection exists if and only if all the following conditions are satisfied:
\begin{enumerate}
\item $v$ is an in-2 out-1 vertex of $G$ that lies on both $P_G(s_1,t_1)$ and $P_G(s_2,t_2)$;
\item if $e_1 = (u_1,v)$, $e_2 = (u_2,v)$ and $e = (v,w)$ are the three edges of $G$ incident
with $v$ in this counter-clockwise order, then either
\begin{itemize}
\item $e_1 \in E(P_G(s_1,t_1))$, $e_2 \in E(P_G(s_2,t_2))$, and $(s_1,t_1) \prec_e (s_2,t_2)$; or
\item $e_2 \in E(P_G(s_1,t_1))$, $e_1 \in E(P_G(s_2,t_2))$, and $(s_2,t_2) \prec_e (s_1,t_1)$.
\end{itemize}
\end{enumerate}
In particular, no two paths $\Pgrid(s_1,t_1)$ and $\Pgrid(s_2,t_2)$ intersect in a vertex
of $\grid(v)$ for an in-1 out-2 vertex $v \in V(G)$.
\end{lemma}
\begin{proof}
First, consider an in-1 out-2 vertex $v \in V(G)$. Then the fact that $\preceq_e$ orders
pairs of $\mathcal{P}(e)$ first according to the destination order and then according to the source
order implies that the cyclic order of the labels of the incoming edges of $v$ in $\Gbunch$
is the reversed cyclic order of the labels of the outgoing edges of $v$ in $\Gbunch$.
These orders stay the same around $\grid(v)$ in $\Ggrid$.
Consequently, no two paths $\Pgrid(s,t)$ intersect at $\grid(v)$.

Consider now an in-2 out-1 vertex $v \in V(G)$ and two paths $P_G(s_1,t_1)$ and $P_G(s_2,t_2)$
passing through $v$. Let $e_1 = (u_1,v)$, $e_2 = (u_2, v)$, and $e = (v, w)$ be the three
edges incident with $v$ in the counter-clockwise order.
Observe that in $\Gbunch$ the counter-clockwise around $v$ 
order of the labels in the bunch of $e_1$ is the restriction of the 
reversed order of the counter-clockwise around $v$ order of the labels in the bunch of $e$. 
Consequently, if $P_G(s_1,t_1)$ and $P_G(s_2,t_2)$ share the same edge incoming to $v$, then
$\Pgrid(s_1,t_1)$ and $\Pgrid(s_2,t_2)$ do not intersect in $\grid(v)$. 
Otherwise, by symmetry assume that $e_1 \in E(P_G(s_1,t_1))$ and $e_2 \in E(P_G(s_2,t_2))$.
Then $(s_1,t_1) \neq (s_2,t_2)$ and we have that $\Pgrid(s_1,t_1)$ and $\Pgrid(s_2,t_2)$ intersect in $\grid(v)$ if and only
if $(s_1,t_1) \prec_e (s_2,t_2)$. This finishes the proof of the lemma.
\end{proof}

\subsection{Canonical solution in a canonical instance}\label{sec:cleaning}

Consider two paths $\Pgrid(s_1,$ $t_1)$ and $\Pgrid(s_2,t_2)$ that intersect at a vertex $x \in \grid(v)$.
To prove that there exists a solution with a nice cactus-like structure, we would like to use the operation of uncrossing at $x$:
replace in the solution the two paths $\Pgrid(s_1,t_1)$ and $\Pgrid(s_2,t_2)$ with $\Pgrid(s_1,t_2)$ and $\Pgrid(s_2,t_1)$, hoping to reduce the number of crossings
by at least $1$ --- the one corresponding to $x$. While such uncrossing is simple to analyze in $G$ or $\Gbunch$, the definition of $\Ggrid$ causes some trouble due to the fact 
that say $\Pgrid(s_1,t_2)$ is not exactly a concatenation of $\Pgrid(s_1,t_1)[s_1,x]$ and $\Pgrid(s_2,t_2)[x, t_2]$, but a ``parallel shift'' of this concatenation.
Luckily, it turns out that nothing bad happens with this ``parallel shift'', but this is not immediate and requires some argumentation. 

The main observation is embedded in the following lemma. Intuitively, it means 
that our canonical paths intersect as little as possible.

\begin{lemma}\label{lem:uncross}
Let $(s_1,t_1)$ and $(s_2,t_2)$ be two distinct elements of $\mathcal{T}$.
For $i=1,2$, let $P_i$ be a path from $s_i$ to $t_i$ in $\Ggrid$ whose projection onto $G$ 
equals $P_G(s_i,t_i)$ (i.e., the projections of $P_i$ and $\Pgrid(s_i,t_i)$ are the same, $P_i$ traverses exactly the same grids
and bunches in the same order as $\Pgrid(s_i,t_i)$).
Furthermore, assume that $P_1$ and $P_2$ do not share
any edge. 
Then the number of intersections of $\Pgrid(s_1,t_1)$ and $\Pgrid(s_2,t_2)$ at nonterminal vertices
(i.e., $|V(\Pgrid(s_1,t_1)) \cap V(\Pgrid(s_2,t_2)) \setminus T|$) is not larger
than the number of transversal intersections of $P_1$ and $P_2$.
\end{lemma}

\begin{proof}
For ease of notation, let $Q_i = \Pgrid(s_i,t_i)$ for $i=1,2$.
We show how to charge every vertex $x \in V(Q_1) \cap V(Q_2) \setminus T$
to a distinct transversal intersection $f(x)$ of $P_1$ and $P_2$.

Fix $x \in V(Q_1) \cap V(Q_2) \setminus T$ and assume $x \in \grid(v)$. 
By Lemma~\ref{lem:grid-cross}, $v$ is an in-2 out-1 vertex with incident edges $e_1$, $e_2$, 
and $e$ in this counter-clockwise order with $e$ being the unique edge with its tail in $v$.
By symmetry, assume that $e_i \in P_G(s_i,t_i)$ for $i=1,2$
and that $(s_1,t_1) \prec_e (s_2,t_2)$. 
However, as $e_i \in P_G(s_i,t_i)$, we have $s_1 \succ_e^\leftarrow s_2$, in particular $s_1 \neq s_2$. 
Hence, $(s_1,t_1) \prec_e (s_2,t_2)$ implies $t_1 \prec_e^\rightarrow t_2$, in particular $t_1 \neq t_2$.

Let $R$ be the maximal subpath of the intersection of $P_G(s_1,t_1)$ and $P_G(s_2,t_2)$ that contains $v$;
$R$ starts at $v$ and ends at a vertex $w \neq v$ ($R$ is of length at least one as
the first edge of $R$ is $e$). Note that $w$ is a nonterminal vertex as $t_1 \neq t_2$.
By the definition of $R$, $w$ is an in-1 out-2 vertex; let $e'$, $e_1'$, and $e_2'$ be the three
edges of $G$ incident with $w$ in this counter-clockwise order with $e'$ being the unique edge
with its head in $w$. Since $t_1 \prec^\rightarrow_e t_2$, it follows that
$t_1 \prec^\rightarrow_{e'} t_2$ and thus $e_i' \in E(P_G(s_i,t_i))$ for $i=1,2$. See Figure~\ref{fig:canonical-cross}.

Since for $i=1,2$ the paths $P_i$ and $Q_i$ share the same projection $P_G(s_i,t_i)$ in $G$,
$P_i$ traverses an edge of the bunch of $e_i$ and an edge of the bunch of $e_i'$.
Therefore, by the assumed counter-clockwise order of the edges around $v$ and $w$,
there exists a transversal intersection of $P_1$ and $P_2$ in $\grid(u)$ for some $u \in V(R)$. 
We denote this intersection by $f(x)$ and we charge $x$ to it.

Lemma~\ref{lem:grid-cross} asserts that $x$ is the only intersection of $V(Q_1)$ and $V(Q_2)$
in all grids $\grid(u)$ for $u \in V(R)$. Therefore our charging scheme is injective
and the lemma is proven.
\end{proof}

\begin{figure}[t]
\begin{center}
\includegraphics[width=0.9\linewidth]{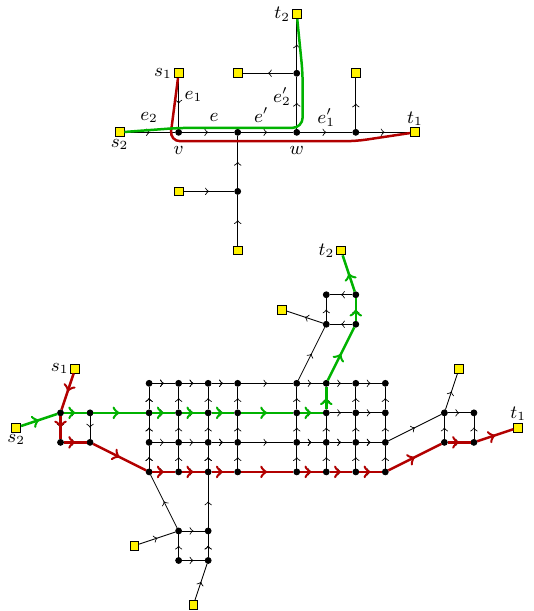}
\end{center}
\caption{The situation in the proof of Lemma~\ref{lem:uncross}. 
  The graph $G$ with paths $P_G(s_i,t_i)$ for $i=1,2$ is at the top
    and the graph $\Ggrid$ with paths $\Pgrid(s_i,t_i)$ for $i=1,2$ is at the bottom.}\label{fig:canonical-cross}
\end{figure}

\begin{corollary}\label{cor:uncross}
Let $\Wgrid_A$ be a solution to \tspname{} on $(\Ggrid,T,\mathcal{P})$ of minimum possible weight 
that visits every edge at most once and visits every nonterminal vertex at most twice. 
Then there exists a solution $\Wgrid_B$ to \tspname{} on $(\Ggrid,T,\mathcal{P})$ also of minimum possible
weight, is canonical, and the number of nonterminal vertices visited more than once on $\Wgrid_B$
is not larger than the number of transversal self-intersections of $\Wgrid_A$.
\end{corollary}
\begin{proof}
Let $W_A$ be the projection of $\Wgrid_A$ onto $G$. Since $\Wgrid_A$ is a solution to $(\Ggrid,T,\mathcal{P})$ of minimum possible weight,
$W_A$ is a solution to $(G,T)$ of minimum possible weight.
In particular, since every edge of $G$ is of positive weight, $W_A$ is locally short.

Let $\Wgrid_B$ be a canonical lift of $W_A$ to $\Ggrid$; that is, 
    if $(t_1,t_2,\ldots,t_{|T|})$ is the witnessing permutation of $W_A$
    then $\Wgrid_B$ is the concatenation of $\Pgrid(t_i,t_{i+1})$ for $1 \leq i \leq |T|$.
Clearly, $\Wgrid_B$ is of the same weight as $W_A$ and $\Wgrid_A$, so it is also a minimum weight
solution to \tspname{} on $(\Ggrid,T,\mathcal{P})$. 

For $1 \leq i \leq |T|$, let $P_i$ be the subwalk from $t_i$ to $t_{i+1}$ on $\Wgrid_A$
and let $Q_i = \Pgrid(t_i,t_{i+1})$. 
Note that since the projection of $P_i$ onto $G$ is $P_G(t_i,t_{i+1})$, $P_i$ is a simple
path in $\Ggrid$. 
From Lemma~\ref{lem:uncross} we infer that for every $1 \leq i < j \leq |T|$
the size of $V(Q_i) \cap V(Q_j) \setminus T$ is not larger than the number
of transversal intersections of $P_i$ and $P_j$. The statement follows.
\end{proof}

As already discussed, there exists a canonical walk $\Wgrid$ in $\Ggrid$ that is a
minimum weight solution to \tspname{} on $(\Ggrid,T,\mathcal{P})$.
Let $\Wgrid$ be such a canonical walk that minimizes the number of self-intersections, that is,
the number of nonterminal vertices that appear on $\Wgrid$ more than once (recall $\Wgrid$ is locally short since it is canonical). 
To finish that $(\Ggrid,T,\mathcal{P})$ is a canonical instance and finish the proof of Lemma~\ref{lem:input-to-blow-up}
it suffices to show that such a minimal $\Wgrid$ is cactuslike.

Assume the contrary; Figure~\ref{fig:cleaning} presents the thought process here.
Let $x \in \grid(v)$ be a nonterminal vertex visited more than once by $\Wgrid$.
Note that $x$ is visited by $\Wgrid$ exactly twice.
Let $W_1$ and $W_2$ be the result of splitting $\Wgrid$ at $x$.
Assume that $W_1$ and $W_2$ intersect at another vertex $x'$. Note that $x'$ needs to be a nonterminal vertex.

Let $\Wgrid_A$ be a closed walk in $\Ggrid$ that is created from $W_1$ and $W_2$ by splitting
them at $x'$. That is, we break $W_1$ and $W_2$ at $x'$ and concatenate them to obtain a single
closed walk $\Wgrid_A$. Note that $\Wgrid_A$ is visits every terminal once and its edge multiset is exactly the same
as the one of $\Wgrid$. In particular, it visits every edge of $\Ggrid$ at most once
and is also a minimum-weight solution to \tspname{} on $(\Ggrid,T,\mathcal{P})$. 

Let $\Wgrid_B$ be a canonical solution to \tspname{} on $(\Ggrid,T,\mathcal{P})$ obtained from Corollary~\ref{cor:uncross}
applied to $\Wgrid_A$.
Corollary~\ref{cor:uncross} asserts that the number of self-intersections of $\Wgrid_B$ is not larger
than the number of transversal self-intersections of $\Wgrid_A$. 
Observe that a nonterminal vertex is visited more than once by $\Wgrid_A$ if and only if it is a self-intersection
of $\Wgrid$ and, furthermore, $x$
is a self-intersection of $\Wgrid$ that is \emph{not} a transversal intersection of $\Wgrid_A$.
Consequently, the number of self-intersections of $\Wgrid$ is strictly larger
than the number of self-intersections of $\Wgrid_B$, contradicting the choice of $\Wgrid$.
This proves Property~\eqref{i:blow-up:canonical} of Definition~\ref{def:blow-up}
and thus finishes the proof of Lemma~\ref{lem:input-to-blow-up}.

\begin{figure}[t]
\begin{center}
\begin{tikzpicture}[scale=0.78]
\node (w1) at (0, 0) {$\Wgrid$};
\node (w2) at (5, 0) {$\Wgrid_A$};
\node (w3) at (10, 0) {$W_A$};
\node (w4) at (15, 0) {$\Wgrid_B$};
\begin{scope}[decoration={markings,mark=at position 0.5 with {\arrow[scale=2]{>}}}]
\draw[postaction={decorate}] (w1) -- node[above] {uncross at $x$ and $x'$}  (w2);
\draw[postaction={decorate}] (w2) -- node[above] {project onto $G$} (w3);
\draw[postaction={decorate}] (w3) -- node[above] {lift canonically} (w4);
\draw[postaction={decorate}, dashed] (w1) to[out=-30, in=-150] node[below,align=center,text width=4.4cm] {all self-intersections stay, $x$ stops to be transversal} (w2);
\draw[postaction={decorate}, dashed] (w2) to[out=-30, in=-150] node[below,align=center,text width=10cm] {Lemma~\ref{lem:uncross}, encapsulated in Corollary~\ref{cor:uncross}: the number of new self-intersections is at most the number of old \emph{transversal} self-intersections} (w4);
\end{scope}
\end{tikzpicture}
\end{center}
\caption{Thought process in Section~\ref{sec:cleaning}.}\label{fig:cleaning}
\end{figure}

\subsection{Properties of a cactuslike walk}\label{sec:reduced}
By Lemma~\ref{lem:input-to-blow-up}, we can concentrate on canonical instances and assume that there is a solution satisfying the properties in Definition~\ref{def:blow-up}. Our goal now is to show that every cactuslike walk can be decomposed into a small number of paths that interact with each other only in a limited way.
Moreover, for future use in the algorithm we will require that in a canonical instance
the paths in the decomposition belong to some small family of candidates.

To formalise the limited interaction between the paths, we need the following definition. 
A pair of paths $P$ and $Q$ are \emph{twisted} if the starting vertex $u$ of $P$ is the ending vertex of $Q$,
  the ending vertex $v$ of $P$ is the starting vertex of $Q$,
  and if $u = w_1,w_2,\ldots,w_\rho = v$ are the vertices of $V(P) \cap V(Q)$ in the order of their appearance
  on $P$, then they appear on $Q$ in the reversed order $v = w_\rho, w_{\rho-1}, \ldots, w_1 = u$. 

We can now state the decomposition lemma. 
We remark that the lemma below would be trivial if we could assume that the walk $W$ does not admit any self-intersections: then breaking $W$ into the subpaths between the terminals would clearly satisfy the conditions.

\begin{lemma}\label{lem:clean-decomposed}
Given a canonical instance $(\Ggrid,T,\mathcal{P})$ one can in polynomial time
compute a family $\mathcal{B}$ of $\Oh(|T|^{12})$ subpaths of canonical paths
such that the following holds. 
Every canonical cactuslike walk $W$ in $(\Ggrid,T,\mathcal{P})$ can be decomposed into $\ell<27|T|$ subpaths $B_1,\ldots,B_\ell$ such that the following conditions are satisfied:
\begin{enumerate}[(a)]
\item\label{p:imp} every path $B_i$ belongs to $\mathcal{B}$;
\item\label{p:cross} for every path $B_i$, either no other path $B_j$ visits an internal vertex of $B_i$,
or there exists a unique other path $B_j$ such that $B_i$ and $B_j$ are twisted;
\item\label{p:few} there are fewer than $27|T|$ self-intersections of $W$ that are not internal vertices of paths $B_i$.
\end{enumerate}
\end{lemma}
\begin{proof}
We initiate $\mathcal{B} = \mathcal{P}$, which is of size $\Oh(|T|^2)$.
The lemma follows trivially if $W$ does not contain any self-intersections, so assume otherwise.

Let $W=(e_1,\ldots,e_p)$. 
As $W$ is canonical, it is locally short; let $(t_1,\ldots,t_{|T|})$ be the witnessing permutation of the terminals, that is, $W$ is the concatenation of the simple paths $P_1,\ldots,P_{|T|}$ such that $P_i = \Pgrid(t_i,t_{i+1})$ for every $i$.
For each $j\in \{1,\ldots,|T|\}$ we choose index $\lambda_j\in \{1,\ldots,p\}$ such that $P_j$ is equal to the subwalk $(e_{\lambda_j+1},\ldots,e_{\lambda_{j+1}})$ of $W$.
We say that $(i,j)$ is a \emph{self-crossing} of $W$ if $i \neq j$ but the head of $e_i$ equals
the head of $e_j$. Note that since $(\Ggrid,T,\mathcal{P})$ is a canonical instance, $W$ visits every vertex at most twice,
every self-crossing happens at a nonterminal vertex that is visited twice by $W$ and corresponds
to a transversal intersection of two paths $P_i$. 

Create an auxiliary graph $H$ on vertex set $x_1,\ldots,x_p$, where $x_i$ can be thought of as a copy of the head of the edge $e_i$ (we also say that $x_i$ {\em{corresponds}} to the head of $e_i$).
In $H$, we put an edge between $x_i$ and $x_{i+1}$ for each $i=1,2,\ldots,p$ (where $x_{p+1}=x_1$), and moreover, for each self-crossing $(i,j)$ of $W$, we put an edge between $x_i$ and $x_j$. 
The latter edges, corresponding to self-crossings, are called {\em{internal}}.
Note that since each terminal is visited exactly once on $W$, vertices $x_{\lambda_1},\ldots,x_{\lambda_{|T|}}$ are the only vertices out of $x_1,\ldots,x_p$ that correspond to terminals.

\begin{claim}\label{cl:Houterplanar}
The graph $H$ is outerplanar and has an outerplanar embedding where the cycle $(x_1,\ldots,x_p)$ is the boundary of the outer face. Moreover, each vertex $x_{\lambda_j}$, for $j\in [|T|]$, has degree $2$ in $H$.
\end{claim}
\begin{proof}
To see that $H$ is a cycle with non-crossing chords, it suffices to show that there are no indices $i<i'<j<j'$ such that both $(i,j)$ and $(i',j')$ are self-crossings of $W$.
However, if this was the case, then the self-crossing $(i',j')$ would yield a crossing of the closed walks $W_1$ and $W_2$ obtained by splitting $W$ at the self-crossing $(i,j)$.
Since $W$ is cactuslike, this cannot happen.

To see that the vertex $x_{\lambda_j}$, corresponding to the terminal $t_j$, has degree $2$ in $H$, observe that otherwise $x_{\lambda_j}$ would be incident to some internal edge of $H$.
This means that $W$ would have a self-crossing at $t_j$, but $W$ visits each terminal at most once; a contradiction.
\cqed\end{proof}

Fix an outerplanar embedding of $H$ as in Claim~\ref{cl:Houterplanar}.
Let $S$ be a graph with vertex set consisting of the inner faces of $H$, where two faces are considered adjacent if and only if they share an edge in $H$.
Since $H$ is outerplanar and connected, it follows that $S$ is a tree.

Consider now any leaf $f$ of $S$. Then the boundary of $f$ consists of one edge of $H$ corresponding to some self-crossing $(i,j)$ of $W$, say at vertex $v$ of $G$, 
and a subpath $Q_f$ of the cycle $(x_1,\ldots,x_p)$ in $H$.
For leaves $f$ of $S$, the subpaths $Q_f$ are pairwise edge disjoint.

\begin{figure}[tb]
\begin{center}
\def\svgwidth{\textwidth}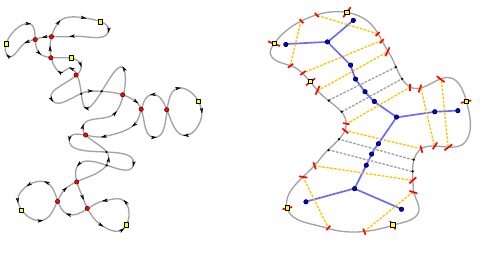
\end{center}
\caption{The original closed walk $W$ (left panel) and the outerplanar graph $H$ constructed based on $W$ (right panel); terminals and vertices visited by $W$ more than once have been named.
Terminals are depicted by yellow squares, the tree $S$ is depicted in blue,
the cycle $(x_1,\ldots,x_p)$ is depicted using solid gray edges, while dashed gray edges are the internal edges of $H$.
Note that the counterclockwise order of vertices on the outer face of $H$ corresponds to the order of visiting corresponding vertices by $W$; for instance, vertices $q,p,o$ appear in this order in $H$, because
they are visited in this order by $W$.
Special edges are colored orange, while red lines depict places where we put dividing points for defining blocks. 
They correspond to vertices depicted by red circles in the left panel, which are in the set $J$.}\label{fig:partition}
\end{figure}

\begin{claim}\label{cl:bound-leaves}
For each leaf $f$ of $S$, the subpath $Q_f$ contains at least one vertex $x_{\lambda_j}$, for some $j\in [|T|]$, as an internal vertex.
Consequently, the tree $S$ has at most $|T|$ leaves.
\end{claim}
\begin{proof}
For the first claim, observe that $Q_f$ corresponds to a closed subwalk $W_f$ of $W$ obtained by splitting $W$ at a self-crossing.
Observe that $W_f$ cannot be entirely contained in any of the paths $P_j$, since $W_f$ visits $v$ twice whereas a simple path
cannot visit any vertex more than once. Hence, $Q_f$ contains some vertex $x_{\lambda_j}$ as an internal vertex.
The second claim follows by noting that paths $Q_f$ are pairwise edge disjoint for different leaves $f$ of $S$, and there are $|T|$ vertices $x_{\lambda_j}$.
\cqed\end{proof}

Observe that in the duality of the outerplanar graph $H$ and the tree $S$, the edges of $S$ are the dual edges of the internal edges of~$H$.
By somehow abusing the notation, we identify each internal edge of $H$ with its dual edge in $S$.

We now define the set of {\em{special}} edges of the tree $S$ as follows.
First, for each vertex $f$ of $S$ of degree at least $3$ in $S$, we mark all edges incident to $f$ special.
Second, for each vertex $x_{\lambda_j}$, for $j\in [|T|]$, we find the unique index $h_j\in [p]$ such that none of vertices $x_{\lambda_j},\ldots,x_{h_j-1}$ is incident to any internal edges of $H$,
but $x_{h_j}$ is incident to such an edge (it exists as we assumed that $W$ has at least one self-intersection).
Then there is a unique special edge of $H$ that is incident both to $x_{h_j}$ and the internal face of $H$ on which $x_{\lambda_j}$ lies (this face is unique since
$x_{\lambda_j}$ has degree $2$ in $H$). We mark this internal edge special as well.

\begin{claim}\label{cl:bound-special}
There are less than $4|T|$ special edges in $S$.
\end{claim}
\begin{proof}
It is well known that in every tree with at most $k$ leaves, the total number of edges incident to vertices of degree at least $3$ is at most $3k-6$. Hence,
since $S$ has at most $|T|$ leaves by Claim~\ref{cl:bound-leaves}, less than $3|T|$ edges of $S$ were marked as special in the first step of marking.
In the second step of marking we mark one edge per each terminal, so the total upper bound of less than $4|T|$ follows.
\cqed\end{proof}

We divide the walk $W=(e_1,\ldots,e_p)$ into {\em{blocks}} as follows.
For any $i\in [p]$, declare $x_i$ a {\em{dividing point}} if either $x_i$ corresponds to a terminal (i.e. $i=\lambda_j$ for some $j\in [|T|]$), or $x_i$ is an endpoint of a special edge.
Then blocks are maximal subwalks of $W$ that do not contain any dividing points as internal vertices.
More precisely, the sequence $(e_{i+1},e_{i},\ldots,e_{i'})$ is a block if both $x_i$ and $x_{i'}$ are dividing points, but none of vertices $x_{i+1},\ldots,x_{i'-1}$ is a dividing point.
It is clear that blocks form a partition of $W$ into less than $9|T|$ subwalks, as there are less than $9|T|$ dividing points by Claim~\ref{cl:bound-special}.
Let $A_1,\ldots,A_r$ be the obtained blocks; we have $r < 9|T|$.
We now establish a number of properties of the obtained blocks.

\begin{claim}\label{cl:bound-cross}
For every block $A_i$, the internal vertices of $A_i$ are not endpoints of any other path $A_j$
and can be internal vertices of at most one other path $A_j$.
\end{claim}
\begin{proof}
Let $D_i$ be the subpath of the cycle $(x_1,\ldots,x_p)$ in $H$ that corresponds to the block $A_i$, for $i=1,\ldots,\ell$.
Note that every intersection of paths $A_i$ and $A_{j}$ at an internal vertex of $A_i$, for $i\neq {j}$, is also a self-crossing of $W$ that corresponds to an internal edges of $H$ that connects
an internal vertex of $D_i$ with a vertex of $D_{j}$.
Fix now some block $A_i$; we will argue that there is at most one other block $A_{j}$ such that $A_i$ and $A_{j}$ intersect at an internal vertex of $A_i$ and every such intersection happens at an internal vertex of $A_{j}$. This would prove the claim.

Let us consider the forest obtained by removing every special edges from the tree $S$. Observe that every connected component of this forest either
\begin{itemize}
\item consists of one vertex being a leaf of $S$, or
\item consists of one vertex of degree at least $3$ in $S$, or
\item is a path (possibly of length $0$) consisting only of vertices of degree $2$ in $S$.
\end{itemize}
This is because any edge incident to a leaf of $S$ is always marked as special by Claim~\ref{cl:bound-leaves}.
By the construction of blocks,  the set of internal faces of $H$ incident to the edges $D_i$ can be spanned by a subtree of $S$
that does not contain any special edge. Consequently, either all the edges of $A_i$ are incident to the same internal face of $H$ (and hence they form an interval on its boundary), or there is a path $R$ in $S$,
consisting only of vertices of degree $2$ connected by non-special edges, such that all the edges of $D_i$ are incident to the faces on this path.
In the former case, $A_i$ does not intersect any other block $A_j$ at an internal vertex of $A_i$, as all internal vertices of $D_i$ have degree $2$ in $H$.
In the latter case, it is easy to see that all the edges of the cycle $(x_1,\ldots,x_p)$ that are incident to some non-endpoint face of $R$ but do not lie on $D_i$, 
are in fact in the same subpath $D_{j}$ for some $j\neq i$.
Then all internal edges of $H$ incident to the internal vertices of $D_i$ have the second endpoint on $D_{j}$, so $A_{j}$ is the only block that may intersect $A_i$ at an internal vertex of $A_i$.
Furthermore, note that since the endpoints of $R$ are of degree $2$ in $S$ and every vertex of $H$ is incident with
at most one internal edge (since every vertex of $\Ggrid$ is visited at most twice by $W$), all 
internal edges of $H$ incident to the internal vertices of $D_i$ have the second endpoint in an internal vertex of $D_{j}$.
\cqed\end{proof}

\begin{claim}\label{cl:bound-few}
There are fewer than $9|T|$ self-intersections of $W$ that are not internal vertices of paths $A_i$.
\end{claim}
\begin{proof}
Observe first that since every nonterminal vertex of $\Ggrid$
lies on at most two canonical paths, every $x_i$ is incident with at most one internal edge. 
Furthermore, self-crossings of $W$ that are not crossings of two distinct blocks are exactly those self-crossings $(i,i')$ for which either $x_i$ or $x_{i'}$ is a dividing point.
Since there are less than $9|T|$ dividing points, the claim follows.
\cqed\end{proof}

Define a set $J \subseteq V(\Ggrid)$ as follows. Start from $J = T$.
Next, consider every quadruple of terminals $s_1,t_1,s_2,t_2$, where $s_1\neq t_1$, $s_2\neq t_2$, $s_1\neq s_2$, and $t_1\neq t_2$
and insert into $J$ the first intersection on $\Pgrid(s_1,t_1)$ of $\Pgrid(s_1,t_1)$ and $\Pgrid(s_2,t_2)$ (if it exists).
Clearly, $|J| \leq |T|^4$.

\begin{claim}\label{cl:get-imp}
Suppose we have indices $1\leq j,j'\leq |T|$, $j\neq j'$. Suppose further on the subpath of $(x_{\lambda_j},x_{\lambda_j+1},\ldots,x_{\lambda_{j+1}})$, vertex $x_{k}$ is the first one that is adjacent in $H$ to any of the vertices
$x_{\lambda_{j'}},x_{\lambda_{j'}+1},\ldots,x_{\lambda_{j'+1}}$ via an internal edge of $H$. Then $x_k$ corresponds to an element of $J$.
\end{claim}
\begin{proof}
It can be easily seen that if $x_k$ corresponds to a vertex $v$, then $v$ is included in the set $J$ when considering the quadruple of terminals $(t_j,t_{j+1},t_{j'},t_{j'+1})$.
\cqed\end{proof}

\begin{figure}[tb]
\begin{center}
\def\svgwidth{0.5\textwidth}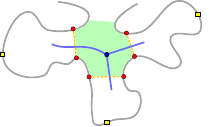
\end{center}
\caption{Situation in the case when $e$ was marked due to being incident to an internal face of $H$ of degree at least $3$ in $S$.}\label{fig:imp-cross}
\end{figure}

\begin{claim}\label{cl:bound-imp}
Every path $A_i$ has both endpoints in $J$. 
\end{claim}
\begin{proof}
We proceed with verification that all the dividing points used in the definition of blocks correspond to vertices of $J$.
This is done explicitly for terminals, so we are left with verifying this for endpoints of special edges.
Suppose that an internal edge $e=(x_i,x_{i'})$ of $H$ is special. Then $x_i$ and $x_{i'}$ correspond to the same vertex $v$ of $G$ such that $(i,i')$ is a self-crossing of $W$ at $v$.
We have two cases, depending on why $e$ was marked as special.

Suppose first that $e$ was marked as special due to being incident to some internal face $f$ of $H$ of degree at least $3$ in $S$; see Figure~\ref{fig:imp-cross}.
This means that in $S$, $f$ has at least two other incident edges, and suppose $e^1$ and $e^2$ are the edge incident to $f$ that are directly preceding and succeeding $e$ in the counter-clockwise order of edges
of $S$ incident to $f$; here, we assume that the cycle $(x_1,\ldots,x_p)$ is oriented counter-clockwise in the plane.
Further, suppose without loss of generality that $e^1,x_i,e,x_{i'},e^2$ are in this counter-clockwise order on the boundary of face $f$.
Now, let $j^1\in [|T|]$ be such that on the subpath $(x_{\lambda_{j^1}},x_{\lambda_{j^1}+1},\ldots,x_i)$ no internal vertex corresponds to a terminal,
and similarly let $j^2\in [|T|]$ be such that on the subpath $(x_{i'},x_{i'+1},\ldots,x_{\lambda_{j^2}})$ no internal vertex corresponds to a terminal.
Observe that since each leaf $f'$ of $S$ has a vertex corresponding to a terminal among internal vertices of $Q_f$ (Claim~\ref{cl:bound-leaves}),
vertices $x_{\lambda_{j^1}}$, $x_{\lambda_{j^2-1}}$, and $x_{\lambda_{j^2}}$ lie on the following parts of the cycle $(x_1,\ldots,x_p)$:
\begin{itemize}
\item denoting $e^1=x_{r_1^1}x_{r^1_2}$, where $x_{r_1^1}$ $x_{r_2^1}$, and $x_i$ lie in this order on $(x_1,\ldots,x_p)$, we have that $x_{\lambda_{j^1}}$ is an internal vertex of $(x_{r_1^1},\ldots,x_i)$;
\item $x_{\lambda_{j^2-1}}$ is an internal vertex of $(x_i,\ldots,x_{i'})$; and
\item denoting $e^2=x_{r_1^2}x_{r^2_2}$, where $x_{i'}$, $x_{r_1^2}$, and $x_{r_2^2}$ lie in this order on $(x_1,\ldots,x_p)$, we have that $x_{\lambda_{j^2}}$ is an internal vertex of $(x_{i'},\ldots,x_{r_2^2})$;
\end{itemize}
In particular, all the vertices  $x_{\lambda_{j^1}}$, $x_{\lambda_{j^2-1}}$, and $x_{\lambda_{j^2}}$ are pairwise different,
and moreover $e$ is the internal edge of $H$ connecting $(x_{\lambda_{j^1}},x_{\lambda_{j^1}+1},\ldots,x_{\lambda_{j^1+1}})$ with $(x_{\lambda_{j^2-1}},$ $x_{\lambda_{j^2-1}+1},\ldots,x_{\lambda_{j^2}})$ that has the earliest possible endpoint on
the former path. The fact that $v \in J$ follows from applying Claim~\ref{cl:get-imp} to $j=j^1$ and $j'=j^2-1$.

Suppose now, without loss of generality, that $e$ was marked special due to the following situation: $i=h_j$ for some terminal $t_j$, and $e$ is the unique edge incident to $x_i$ that is also incident to the
internal face $f$ of $H$ on whose boundary lies $x_{\lambda_j}$.
Then on the subpath $(x_{\lambda_j},x_{\lambda_{j}+1},\ldots,x_i)$, all vertices have degree $2$ in $H$, apart from $x_i$ itself, so in particular they are not incident to any internal edge of $H$.
Suppose now that $j'\in [|T|]$ is such that on the subpath $(x_{\lambda_{j'}},x_{\lambda_{j'}+1},\ldots,x_{i'})$ no internal vertex corresponds to a terminal. By Claim~\ref{cl:bound-leaves} it is easy to see that
$j\neq j'$. Moreover, from the previous observation it follows that $x_i$ is the earliest vertex on $(x_{\lambda_j},x_{\lambda_{j}+1},\ldots,x_{\lambda_{j+1}})$ that is adjacent to any vertex of 
$(x_{\lambda_{j'}},x_{\lambda_{j'}+1},\ldots,x_{\lambda_{j'+1}})$ via an internal edge of $H$, because the earlier vertices were not incident to any internal edges at all.
The fact that $v \in J$ follows from applying Claim~\ref{cl:get-imp} to $j$ and $j'$.
\cqed\end{proof}

\begin{claim}\label{cl:compute-J}
One can compute a family $\mathcal{A}$ of $\Oh(|T|^6)$ subpaths of paths in $\mathcal{P}$
that contains all blocks $A_i$. 
\end{claim}
\begin{proof}
Recall that each path $A_i$ is a subpath of a path $\Pgrid(s,t)$ with endpoints in $J$.
Since every nonterminal vertex in $\Ggrid$ participates in at most two canonical paths, an element of $J$
on $\Pgrid(s,t)$ is either an endpoint or an intersection with some other canonical path $\Pgrid(s',t')$ that is
either first on $\Pgrid(s,t)$ or first on $\Pgrid(s',t')$. There are at most $|T|^2$ choices for $(s,t)$
and, given $(s,t)$, $\Oh(|T|^2)$ choices for each of the endpoints of a path $A_i$.
The claim follows. 
\cqed\end{proof}

Having established the properties of the block $A_i$, we now show how to partition them into the desired paths $B_i$. 
If a block $A_i$ does not intersect any other block $A_j$ at an internal vertex, we leave $A_i$ untouched.

Consider a pair of blocks $A_i$ and $A_{j}$ that intersect at a vertex that is internal to both $A_i$ and $A_{j}$
(cf. Claim~\ref{cl:bound-cross}). 
Let $A_i$ be a subpath of $\Pgrid(s_1,t_1)$ and let $A_{j}$ be a subpath of $\Pgrid(s_2,t_2)$. 

Let $x$ and $y$ be two intersections of $A_i$ and $A_{j}$ such that $y$ is later on $A_i$ than $x$. 
By Lemma~\ref{lem:grid-cross}, $x \in \grid(v)$ and $y \in \grid(u)$ for two distinct vertices $v$ and $u$ of $G$.
If $y$ is also later on $A_{j}$ than $x$, then the uniqueness of shortest paths in $G$ implies that
the intersection of $P_G(s_1,t_1)$ and $P_G(s_2,t_2)$ contains $P_G(v,u)$. This is a contradiction
with Lemma~\ref{lem:grid-cross} and the fact that $A_i$ and $A_{j}$ intersect in $y \in \grid(u)$.

Hence, if $x_1,x_2,\ldots,x_\rho$ are the intersections of $A_i$ and $A_{j}$ in the order of their appearance on $A_i$,
then they appear in $A_{j}$ in the reversed order $x_\rho, x_{\rho-1}, \ldots, x_1$. 
We split $A_i$ and $A_{j}$ into three paths each (two if $\rho = 1$) at $x_1$ and $x_\rho$. 
Observe that if $\rho > 1$, then the middle parts of $A_i$ and $A_{j}$ are twisted. See Figure~\ref{fig:untwist}.

Let $B_1,B_2,\ldots,B_\ell$ be the obtained paths (i.e., all blocks $A_i$ that do not intersect any other block at an internal
vertex and the at most three subpaths of a block $A_i$ obtained as above otherwise). 
Since every block $A_i$ is split into at most three paths $B_i$ and there are less than $9|T|$ block, we have $\ell < 27|T|$. 
Also, Claim~\ref{cl:bound-few} immediately implies Point~\ref{p:few}
while Claim~\ref{cl:bound-cross} with the construction above implies Point~\ref{p:cross}.

Finally, to compute the set $\mathcal{B}$ that contains all paths $B_i$, proceed as follows.
Compute the family $\mathcal{A}$ from Claim~\ref{cl:compute-J} and initiate $\mathcal{B} = \mathcal{A}$.
Then, for every pair $A \in \mathcal{A}$ and every $A' \in \mathcal{A}$ that intersects $A$ at an internal vertex,
split $A$ into at most three parts at the first and last intersection with $A'$ (first and last refer to the order on $A$)
and insert the parts into $\mathcal{B}$. Clearly, $|\mathcal{B}| = \Oh(|T|^{12})$ and $\mathcal{B}$ contains all paths $B_i$.
\end{proof}

\begin{figure}[tbp]
\begin{center}
\includegraphics[scale=1]{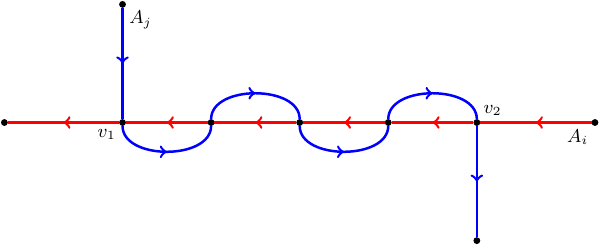}
\end{center}
\caption{Construction of paths $B_i$ at the end of the proof of Lemma~\ref{lem:clean-decomposed}.
  For every two blocks $A_i$ and $A_j$ that intersect (red and blue paths in the figure), we split
    each of them into at most three subpaths at the first and last intersection.}
      \label{fig:untwist}
\end{figure}

\subsection{Enumerating subsets of a walk}\label{sec:enum}
\newcommand{\goodfam}{\mathcal{F}}

Let $(\Ggrid,T,\mathcal{P})$ be a canonical instance. 
Our main technical result, proved in this section, is that any canonical cactuslike walk can be hierarchically decomposed using closed curves of ``complexity'' $|T|^{\Oh(\sqrt{|T|})}$. 
We first formalize what we mean by a decomposition.

\begin{definition}
Let $\alpha > 0$ be a fixed constant.
Let $(G,T)$ be a \tspname{} instance. Let $W$ be a walk that visits every terminal exactly once and let $\pi_W = (t_1,t_2,\ldots,t_{|T|})$ be a witnessing permutation.
A set $A \subseteq T$ is an \emph{$\alpha$-good section} of $(W,\pi_W)$ if $A$ can be partitioned into at most $\alpha \sqrt{|T|}$ subsets that form contiguous subsequences of $\pi_W$.

An \emph{$\alpha$-good decomposition} of $W$ and $\pi_W$ is a pair $(\tree, \beta)$ where $\tree$ is a rooted binary tree and $\beta \colon V(\tree) \to 2^T$ is a function with the following properties:
\begin{enumerate}[(1)]
\item\label{p:good} $\beta(v)$ is an $\alpha$-good section of $(W,\pi_W)$ for every $v \in V(\tree)$;
\item\label{p:root} $\beta(r) = T$ for the root $r$ of $\tree$;
\item\label{p:chld} every non-leaf node $v$ of $\tree$ has two children $v_1,v_2$ with $\beta(v_1) \cap \beta(v_2) = \emptyset$, $\beta(v) = \beta(v_1) \cup \beta(v_2)$;
\item\label{p:leaf} every leaf node $v$ of $\tree$ satisfies $|\beta(v)| \leq \alpha \sqrt{|T|}$.
\end{enumerate}
\end{definition}
Note that both $T$ and every set $A \subseteq T$ of size at most $\alpha \sqrt{|T|}$ is always a good section, regardless of the choice of $W$ and $\pi_W$.

The following section shows that if an optimum solution admits a good decomposition where every set $\beta(s)$ belongs to a known family $\mathcal{F}$, then an optimum solution can be computed efficiently using a dynamic programming algorithm. 
The main result of this section shows that existence of such a good decomposition and family $\mathcal{F}$:
\begin{lemma}\label{lem:enum}
There exists a universal constant $C > 0$ such that the following holds.
Given a canonical instance $(\Ggrid,T,\mathcal{P})$,
one can in time $|T|^{\Oh(\sqrt{|T|})} n^{\Oh(1)}$ compute a family $\goodfam \subseteq 2^T$ of size $|T|^{\Oh(\sqrt{|T|})}$ such that
for every canonical cactuslike walk $W$ and its witnessing permutation $\pi_W$, there exists a $C$-good decomposition $(\tree,\beta)$ of $(W,\pi_W)$
such that every set $\beta(s)$ for $s \in V(\tree)$ belongs to $\goodfam$.
\end{lemma}

The rest of this section is devoted to the proof of Lemma~\ref{lem:enum}. 
Fix the walk $W$ as in the statement.

For every terminal $t$ fix a face $f_t$ incident with $t$. 
Let $D_0$ be a minimal tree in the dual of $\Ggrid$ that spans all faces $f_t$. 
We augment the graph $\Ggrid$ by adding $D_0$ to it as follows.
First, we add $V(D_0)$ to $\Ggrid$. For every edge $f_1f_2 \in E(D_0)$ that crosses an edge $e$ of $\Ggrid$,
we subdivide $e$ with a vertex $z_e$ (distributing the weight of $e$ among two parts arbitrarily) and add arcs $(f_1,z_e)$ and $(f_2,z_e)$ of weight $+\infty$ each.
Finally, for every $t \in T$, we add an arc $(f_t,t)$ of weight $+\infty$
and proclaim $\Tedges = \{(f_t,t)~|~t \in T\}$ the set of terminal edges.
By sligtly abusing the notation, we keep the name $\Ggrid$ for the modified graph. Let $D$ be the subgraph of $\Ggrid$ consisting of all edges of weight $+\infty$;
note that $D$ (without directions of arcs, which are in fact irrelevant) is a tree spanning all terminals and every terminal is a leaf of $D$. 
Finally, by interpreting the canonical paths to the modified $\Ggrid$ in the natural way, we have that $(\Ggrid,T,\mathcal{P})$ is still a canonical instance;
here the essential observation is that the new edges of weight $+\infty$ do not change the structure of the shortest paths in $\Ggrid$.
Intuitively, the purpose of $D$ is to control the homotopy types of closed curves in the plane punctured at the terminals, by examining how they cross with $D$.

We apply Lemma~\ref{lem:clean-decomposed}, obtaining a family $\mathcal{B}$ of $\Oh(|T|^{12})$ subpaths
of the paths in $\mathcal{P}$ such that the walk $W$ can be decomposed into $\ell = \Oh(|T|)$
paths $B_1,B_2,\ldots,B_\ell$, all belonging to $\mathcal{B}$.
Let $\Qq = \{B_1,B_2,\ldots,B_\ell\}$.

\paragraph{Graphs $H$, $H^\times$, and an sc-branch decomposition.}
We define a subgraph $H$ of $\Ggrid$ as the union of $D$
and all paths from $\Qq$.

Although $H$ is a plane graph, it can have an unbounded number of vertices and potentially large branchwidth. 
Let $H^\times$ be the graph obtained from $H$ by contracting, for every $Q \in \Qq$, all internal vertices of $Q$ into one vertex $u_Q$.
Thus, $Q$ gets contracted into a path $Q^\times$ consisting of two edges and three vertices: the former
endpoints and $u_Q$.
Recall that since the paths of $\Qq$ are vertex-disjoint except for possibly having common endpoints and pairs of twisted paths, the contractions on different paths $Q\in \Qq$ do not interfere with each other, except for the pair of twisted paths. However, if $Q_1$ and $Q_2$ are twisted,
then $u_{Q_1} = u_{Q_2}$ and $Q_1^\times$ and $Q_2^\times$ are two paths on the same three vertices.

Furthermore, since we contract only edges of paths of $\Qq \subseteq \mathcal{B}$, we do not contract any edge of $D$. The edges of $D$ are still present in $H^\times$, but $D$ may no longer be a tree.

We have the following bound.
\begin{claim}\label{cl:Hx-decomp}
The graph $H^\times$ admits an sc-branch decomposition $(\tree, \brleaf, \brnoose)$ of width $\Oh(\sqrt{|T|})$.
\end{claim}
\begin{proof}
First, note that $H^\times$ is connected.
By Theorems~\ref{thm:planar-branch} and~\ref{thm:scbranch}, it suffices to show only that $H^\times$ has $\Oh(|T|)$ vertices of degree at least $3$.
To this end, note that every vertex of $H^\times$ of degree at least 3 is either a vertex of $D$ of degree at least 3 or one of the three vertices of the contracted path $Q^\times$ for some $Q \in \Qq$
(note that every terminal falls into the latter case).
The claim follows from the fact that $|\Qq| = \Oh(|T|)$ and that $D$ has less than $|T|$
vertices of degree at least $3$.
\cqed\end{proof}
Let $(\tree,\brleaf,\brnoose)$ be the sc-branch decomposition of $H^\times$ given by Claim~\ref{cl:Hx-decomp}. Our goal is to show that this sc-branch decomposition can be turned into an appropriate good decomposition of the solution $W$. For this purpose, for every noose $\gamma$ appearing in the sc-branch decomposition, we have to show two main properties:
\begin{enumerate}
\item[(1)] The terminals inside $\gamma$ appear on $O(\sqrt{|T|})$ contigious subseqences of the solution $W$ (see Claims \ref{cl:gamma-touched} and \ref{cl:good}).
  \item[(2)] We can compute a set $\mathcal{F}$ such that the set of terminals inside $\gamma$ appears in $\mathcal{F}$ (see Claim~\ref{cl:gamma-enum}).
\end{enumerate}
After establishing these two properties, Lemma~\ref{lem:enum} follows in a straightforward way.
 
\paragraph{Branch decomposition gives good decomposition.}
We now show that a good decomposition of $W$ can be inferred from the sc-branch decomposition $(\tree, \brleaf, \brnoose)$.
Root the tree $\tree$ at an arbitrary leaf $r$ such that $\brleaf(r)$ is not of the form $(f_t,t)$ for a terminal $t$
and define $\beta(r) = T$.
For every node $s \in V(\tree) \setminus \{r\}$ with a parent edge $f$, 
we define $\beta(s)$ as follows. Let $\{\Tedges_1,\Tedges_2\}$ be the partition of the terminal edges induced by $\brnoose(f)$
and assume that the side of $\brnoose(f)$ that contains $\brleaf(r)$ corresponds to the set $\Tedges_2$.
Then, we put $\beta(s) = \{t \in T~|~(f_t,t) \in \Tedges_1\}$.

We now verify that $(\tree,\beta)$ is a good decomposition of $W$ (formally, after removing the root in order to have a binary tree).

Let $\gamma$ be a noose with respect to $H^\times$. A path $Q \in \Qq$ is \emph{touched}
by $\gamma$ if $\gamma$ visits $u_Q$ or one of the endpoints of $Q$. 
We denote the set of paths touched by $\gamma$ by $\Qq_\gamma \subseteq \Qq$ and observe the following.
\begin{claim}\label{cl:gamma-touched}
If a noose $\gamma$ with respect to $H^\times$ visits at most $r$ vertices of $H^\times$,
   then there are at most $4r$ paths touched by $\gamma$.
\end{claim}
\begin{proof}
Recall that every vertex of $\Ggrid$ lies on at most two canonical paths, and hence $W$
visits every vertex of $\Ggrid$ at most twice. Consequently, every vertex $v$ of $H^\times$
that is not of the form $u_{Q_0}$ for some $Q_0 \in \Qq$ lies on at most four paths of $\Qq$:
for every of the at most two visits of $W$ in $v$, $v$ is either an internal vertex of some path in $\Qq$
or an ending vertex of one path in $\Qq$ and a starting vertex of another path in $\Qq$.
Furthermore, a vertex $u_{Q_0}$ for some $Q_0 \in \Qq$ lies on exactly two paths $Q^\times$ 
if $Q_0$ is twisted with another path, and on only one path $Q_0^\times$ otherwise.
We infer that $|\Qq_\gamma| \leq 4r$, as desired.
\cqed\end{proof}

With Claim~\ref{cl:gamma-touched} at hand, we can now verify the properties of a good decomposition.
\begin{claim}\label{cl:good}
$(\tree-r, \beta)$ is a $C$-good decomposition of $W$ and $\pi_W$ for sufficiently
large universal constant $C$.
\end{claim}
\begin{proof}
We start with Property~\ref{p:good}. 
Consider a noose $\brnoose := \brnoose(f)$ for some edge $f$ of $\tree$. 
Let $\{\Tedges_1,\Tedges_2\}$ be the partition of $\Tedges$ induced
by $\brnoose$ and $T_i = \{t \in T~|~(f_t,t) \in \Tedges_i\}$ for $i=1,2$.

By the properties of the sc-branch decomposition $(\tree, \brleaf, \brnoose)$, $\brnoose$ visits at most $\alpha \sqrt{|T|}$ vertices of $H^\times$ for some universal constant $\alpha$.
Recall that the walk $W$ is partitioned into paths $B_1,B_2,\ldots,B_\ell$. 
Claim~\ref{cl:gamma-touched} asserts that at most $4\alpha\sqrt{|T|}$ paths $B_i$ are touched by $\brnoose$.

The removal of the paths of $\Qq_\brnoose$ from the walk $W$ splits $W$ into 
at most $4\alpha\sqrt{|T|}$ subwalks. By the definition of $\Qq_\brnoose$, the set of terminals
visited by each such subwalk either lies on one side of $\brnoose$ or is a single
terminal lying on $\brnoose$. Consequently, the set of terminals visited by each such subwalk
is either fully contained in $T_1$ or fully contained in $T_2$. Property~\ref{p:good} follows from the definition of $\beta$.

Let us now verify the remaining properties of a good decomposition one-by-one.
We have $\beta(r) = T$ by definition, and note that by the choice of $\brleaf(r)$ we have $\beta(r') = T$ for the unique child $r'$ of $r$ in $\tree$.
This ensures property~\ref{p:root}.
For property~\ref{p:chld}, pick a non-leaf non-root node $s$ of $\tree$, and observe that it has always exactly two children, say $s_1$ and $s_2$. 
Let $f$ be the edge of $\tree$ connecting $s$ with its parent, and let $f_1,f_2$ be edges connecting $s$ with $s_1,s_2$, respectively.
By the properties of a branch decomposition, the set of edges on the side of $\brnoose(f)$ that does not contain $\brleaf(r)$ is partitioned into sets defined in the same manner for $\brnoose(f_1)$ and 
$\brnoose(f_2)$. As $\beta(s),\beta(s_1),\beta(s_2)$ are defined by including every terminal $t$ depending on whether the edges $(f_t,t)$ is included in the sets above, property~\ref{p:chld} follows.
Finally, for property~\ref{p:leaf}, note that for a leaf $s$ with parent edge $f$, the noose $\brnoose(f)$ encloses a single edge, and thus $|\beta(s)| \leq 1$.
\cqed\end{proof}

Thus, our goal now is to construct a small family of subsets of $T$ that contains all sets $\beta(s)$
for nonroot nodes $s$ of $\tree$.

\paragraph{Enumeration algorithm.}
Intuitively, the partition of the terminal edges induced by a noose $\brnoose=\brnoose(f)$ for some $f \in E(\tree)$ may be guessed as follows.
By Claim~\ref{cl:gamma-touched}, we have $|\Qq_\brnoose| = \Oh(\sqrt{|T|})$.
Furthermore, $\Qq_\brnoose \subseteq \Qq \subseteq \mathcal{B}$ while
the set $\mathcal{B}$ is known to the algorithm and of size $\Oh(|T|^{12})$.
Consequently, we can guess the set $\Qq_\brnoose$
and reconstruct the parts of $H^\times$ visited by $\brnoose$. 
This will be sufficient to invoke Lemma~\ref{lem:noose-enum}.

\begin{claim}\label{cl:gamma-enum}
In time $|T|^{\Oh(\sqrt{|T|})} n^{\Oh(1)}$ one can enumerate a family $\Aa$ of $|T|^{\Oh(\sqrt{|T|})}$ subsets of $T$ 
such that for every nonroot vertex $s \in V(\tree)$ we have $\beta(s) \in \Aa$.
\end{claim}
\begin{proof}
Consider a nonroot vertex $s \in V(\tree)$ with parent edge $f$
and noose $\brnoose := \brnoose(f)$.
Claim~\ref{cl:gamma-touched} ensures that $|\Qq_\brnoose| = \Oh(\sqrt{|T|})$
while $\Qq_\brnoose \subseteq \Qq \subseteq \mathcal{B}$ and $|\mathcal{B}| = \Oh(|T|^{12})$.
We branch into $|T|^{\Oh(\sqrt{|T|})}$ subcases, guessing (considering all possible options)
for the set $\Qq_\brnoose$.

Given $\Qq_\brnoose$, construct the graph $H_\brnoose$ as the union of $D$ and all paths
from $\Qq_\brnoose$. Note that $H_\brnoose$ is a subgraph of $H$ that contains all terminal edges. 
Construct $H_\brnoose^\times$ from $H_\brnoose$ in the same way we constructed $H^\times$ from
$H$: for every $Q \in \Qq_\brnoose$, contract $Q$ into a three-vertex path $Q^\times$
consisting of the endpoints of $Q$ and an internal vertex $u_Q$. 

The graph $H_\brnoose^\times$ is not necessarily a subgraph of $H^\times$, but it contains
all terminal edges and, since we added all paths touched by $\brnoose$ to $H_\brnoose$
and contracted them while constructing $H_\brnoose^\times$, the curve $\brnoose$
(after some possible shifts within the faces of $H_\brnoose^\times$ to accommodate
 differences in planar drawings of $H_\brnoose^\times$ and $H^\times$)
is a noose with respect to $H_\brnoose^\times$ that visits $\Oh(\sqrt{|T|})$ vertices
of $H_\brnoose^\times$ and partitions $\Tedges$ in the same way as in $H^\times$.

However, the graph $H_\brnoose^\times$ may not be connected, so we cannot use Lemma~\ref{lem:noose-enum} directly. Instead, we crucially use now the fact that $H_\brnoose$ contains the tree $D$
that is connected and contains all terminal edges.
We have that $H_\brnoose$ contains a connected component $C_\brnoose$ that contains $D$
and, consequently, $H_\brnoose^\times$ contains a connected component $C_\brnoose^\times$
that contains the image of $C_\brnoose$. In particular, $C_\brnoose^\times$ contains
all terminal edges. Hence, to understand how $\gamma$ partitions terminal edges, 
it suffices to apply Lemma~\ref{lem:noose-enum} to $C_\brnoose^\times$, instead of the entire
graph $H_\brnoose^\times$.

More precisely, by applying Lemma~\ref{lem:noose-enum}
to $C_\brnoose^\times$, $\Tedges$, and $\ell = \Oh(\sqrt{|T|})$, we obtain
a family $\Aa(\Qq_\brnoose)$ of $|T|^{\Oh(\sqrt{|T|})}$ partitions of
$\Tedges$ that contains the one induced by $\brnoose$.

Consequently, for every of the $T^{\Oh(\sqrt{|T|})}$ guesses of $\Qq_\brnoose$, for every
partition $\{\Tedges_1,\Tedges_2\} $ $\in \Aa(\Qq_\brnoose)$ and every $i=1,2$,
we may output $\{t\in T~|~(f_t,t) \in \Tedges_i\}$ as an element of $\Aa$ and conclude.
\cqed\end{proof}

Claims~\ref{cl:good} and~\ref{cl:gamma-enum} conclude the proof of Lemma~\ref{lem:enum}.

\subsection{Dynamic programming algorithm}\label{sec:dp}
In this section we show that, given the family $\goodfam$ obtained using Lemma~\ref{lem:enum},
one can find a shortest walk visiting all terminals in time $|\goodfam|^{\Oh(1)} \cdot |T|^{\Oh(\sqrt{|T|})} \cdot n^{\Oh(1)}$
by a standard dynamic programming approach.
More formally, we show the following lemma.

\begin{lemma}\label{lem:dp}
Given a canonical instance $(\Ggrid,T,\mathcal{P})$ and a family $\goodfam$ of subsets of $T$,
one can in time $|\goodfam|^{\Oh(1)} \cdot |T|^{\Oh(\sqrt{|T|})} \cdot n^{\Oh(1)}$ compute a canonical walk $W_0$ of total length not greater than
the minimum length of a canonical walk $W$ for which there exists a good decomposition $(\tree,\beta)$ 
satisfying $\{\beta(s) \colon s \in V(\tree)\} \subseteq \goodfam$.
\end{lemma}
\begin{proof}
Without loss of generality we assume that $\goodfam$ contains all subsets $A\subseteq T$ with $|A|\leq C\sqrt{|T|}$, where the constant $C$ comes from Lemma~\ref{lem:enum}; this is because the number of such subsets is $|T|^{\Oh(\sqrt{|T|})}$, 
so we may just add them to $\goodfam$.

A \emph{state} consists of a set $A \in \goodfam$ and a family $\mathcal{M}$ of $\Oh(\sqrt{|T|})$ ordered pairs of (not necessarily different) terminals from $A$. 
Note that there are  $|\goodfam| \cdot |T|^{\Oh(\sqrt{|T|})}$ states.

A \emph{realization} of a state $(A, \mathcal{M})$ is a mapping $P$ that
assigns to every pair $(t,t') \in \mathcal{M}$ a walk $P(t,t')$ from $t$ to $t'$ in $\Ggrid$
in such a manner that the walks $\{P(t,t') \colon (t,t') \in \mathcal{M}\}$ together visit all terminals of $A$.
The {\em{weight}} of a realization is the sum of the weights of all walks in it.
In our dynamic programming algorithm we shall compute a realization $P_{(A,\mathcal{M})}$ for every
state $(A, \mathcal{M})$, in the order of increasing size of $A$.

Given two walks $Q_1$ and $Q_2$ with endpoints in $T$, a \emph{concatenation} of $Q_1$ and $Q_2$ is a walk consisting of the walk $Q_1$, then the canonical path from the ending point of $Q_1$ to the starting vertex of $Q_2$, and 
then the walk $Q_2$.
This definition naturally generalizes to concatenations of longer sequences of walks.

For states with $|A| \leq C\sqrt{|T|}$, we compute a minimum weight realization by brute force,
as there are $|T|^{\Oh(\sqrt{|T|})}$ ways to arrange $A$ into a set of sequences corresponding to
the terminals visited by different paths of the realization of minimum weight.

For states $(A,\mathcal{M})$ with 
larger sets $A$, we iterate over all partitions of the form $A = A_1 \uplus A_2$ with $A_1,A_2 \in \goodfam$ and $|A_1|,|A_2|<|A|$,
and all states $(A_1,\mathcal{M}_1)$ and $(A_2,\mathcal{M}_2)$ with precomputed realizations $P_1$ and $P_2$,
respectively.
We iterate over all possibilities of concatenating paths from the images of 
$P_1$ and $P_2$ by brute force.
More formally, we iterate over all possible families $\mathcal{Z}$ of sequences
of elements of $\mathcal{M}_1$ and $\mathcal{M}_2$ that uses every pair from $\mathcal{M}_1$ and
$\mathcal{M}_2$ exactly once. 
Since $|\mathcal{M}_1|,|\mathcal{M}_2| = \Oh(\sqrt{|T|})$,
there are $|T|^{\Oh(\sqrt{|T|})}$ choices for the family $\mathcal{Z}$. 
For every such family $\mathcal{Z}$, we construct a mapping $P$ as follows:
for every sequence $((s_1,t_1),(s_2,t_2),\ldots,(s_r,t_r)) \in \mathcal{Z}$
we make the concatenation of the walks $(P_1 \cup P_2)(s_i, t_i)$ for $1 \leq i \leq r$
and let $P$ map $(s_1,t_r)$ to this concatenation.
If in the end $P$ has $\mathcal{M}$ as a domain, we consider $P$ as a candidate 
realization of $(A,\mathcal{M})$ and finally choose a realization of minimum weight among
all choices of $(A_1,\mathcal{M}_1)$, $(A_2,\mathcal{M}_2)$, and $\mathcal{Z}$.

Finally, we iterate over all states $(T, \{(t,t')\})$ for terminals $t,t' \in T$ and
set $\pi_{t,t'}$ to be the order in which $P_{(T,\{(t,t')\})}(t,t')$ traverses the terminals.
For each such choice, compute a canonical walk $W_{t,t'}$ with the witnessing permutation $\pi_{t,t'}$,
and return the minimum-weight walk found.

Clearly, the algorithm returns a canonical walk. Consider a canonical walk $W$ for which 
there exists a witnessing permutation $\pi_W$ and a good decomposition $(\tree,\beta)$
such that $\{\beta(v) \colon v \in V(T)\} \subseteq \goodfam$.
From the definition of a good decomposition, for every node $s \in V(\tree)$ there exists
a collection $\mathcal{P}_s$ of at most $C\sqrt{|T|}$ subwalks of $W$ that visit exactly
the terminals of $\beta(s)$. Furthermore, we can choose these collections in such a manner
that every subwalk starts and ends at a terminal, 
for the root $r$ the collection $\mathcal{P}_r$ consists of a single subwalk of $W$ from the first to the last terminal of $\pi_W$,
and for every node $s$ with children $s^1$ and $s^2$, the walks of $\mathcal{P}_s$ are concatenations of some walks of $\mathcal{P}_{s^1}$ and $\mathcal{P}_{s^2}$,
where every walk in $\mathcal{P}_{s^1}$ and $\mathcal{P}_{s^2}$ is used exactly once.
Let $\mathcal{M}_s$ be the family of pairs of endpoints of $\mathcal{P}_s$.
Then a standard inductive argument shows that the realization
for $(\beta(s),\mathcal{M}_s)$ is of weight at most the total weight of the walks in $\mathcal{P}_s$; to ensure the correctness of computation for states with $|A|\leq C\sqrt{|T|}$ we use the assumption that $\goodfam$
contains all subsets $A$ satisfying this condition.
Consequently, if $t,t'$ are the first and the last terminal on $\pi_W$, then
the computed realization of $(T,\{(t,t')\})$ is of weight at most the weight of the subwalk of $W$ from $t$ to $t'$.
Hence the canonical walk computed for $\pi_{t,t'}$ is of weight not larger than the weight of $W$, which concludes the proof.
\end{proof}

By pipelining Lemma~\ref{lem:enum} with Lemma~\ref{lem:dp} we obtain Lemma~\ref{lem:alg-blow-up}.

\subsection{Wrap up}

We conclude the proof of Theorem~\ref{thm:dirtsp}.
Let $(G,T)$ be an input instance.
By Lemma~\ref{lem:input-to-blow-up} we obtain an equivalent canonical instance $(\Ggrid,T,\mathcal{P})$.
Then Lemma~\ref{lem:alg-blow-up} allows us to find a minimum-weight solution to \tspname{} in $(\Ggrid,T,\mathcal{P})$ in time $|T|^{\Oh(\sqrt{|T|})} n^{\Oh(1)}$
which can be projected back to a minimum-weight solution to \tspname{} in $(G,T)$.
This concludes the proof of Theorem~\ref{thm:dirtsp}.

\bibliographystyle{plainurl}

\bibliography{references}

\end{document}